\documentclass[twocolumn,reprint,aps,prl, floatfix,amsmath,amssymb,superscriptaddress]{revtex4-2}

\usepackage{graphicx}% Include figure files
\usepackage{bm}% bold math
\usepackage{braket}
\usepackage{xcolor} 
\usepackage{subfigure}     % This is supposed to be obsolete, superseded by subfig, but subfig conflicts with revtex
\usepackage{amsthm}
\usepackage{hyperref}% add hypertext capabilities
\usepackage{CJK} % display Chinese, Korean or Japanese names

\newtheorem{theorem}{Theorem}
\newtheorem{definition}[theorem]{Definition}
\newtheorem{lemma}[theorem]{Lemma}
\newtheorem{proposition}[theorem]{Proposition}

\newtheorem{porism}[theorem]{Porism}

\newtheorem{fact}[theorem]{Fact}
\newcommand{\bigket}[1]{ {\big|{#1}\big\rangle} }

\newcommand{\CFT}{{\textsl{\tiny CFT}}}
\newcommand{\calH}{\mathcal{H}}
\newcommand{\calHconj}{\mathord{\,\overline{\!\mathcal{H}}}}
\DeclareMathOperator{\rank}{rank}
\DeclareMathOperator{\Tr}{Tr}

\definecolor{dkgreen}{rgb}{0,0.6,0}

\definecolor{purple}{rgb}{0.5,0,0.5}

\begin{document}

\title{Universal tripartite entanglement in one-dimensional many-body systems}% Force line breaks with \\
\date{\today}% It is always \today, today,
\author{Yijian Zou}
\affiliation{Perimeter Institute for Theoretical Physics, Waterloo ON, N2L 2Y5, Canada}
\affiliation{University of Waterloo, Waterloo ON, N2L 3G1, Canada}
\affiliation{Sandbox@Alphabet, Mountain View, CA 94043, USA}
\author{Karthik Siva}
\affiliation{Department of Physics, University of California, Berkeley, CA 94720, USA}
\author{Tomohiro Soejima}%
\affiliation{Department of Physics, University of California, Berkeley, CA 94720, USA}
\author{Roger S. K. Mong}
\affiliation{Department of Physics and Astronomy, University of Pittsburgh, Pittsburgh, PA 15260, USA}
\author{Michael P. Zaletel}
\affiliation{Department of Physics, University of California, Berkeley, CA 94720, USA}
	\affiliation{Materials Sciences Division, Lawrence Berkeley National Laboratory, Berkeley, California 94720, USA
}

\begin{abstract}
Motivated by conjectures in holography relating the entanglement of purification and  reflected entropy to the entanglement wedge cross-section, we introduce two related non-negative measures of tripartite entanglement $g$ and $h$. We prove structure theorems which show that states with nonzero $g$ or $h$ have nontrivial tripartite entanglement. We then establish that in 1D these tripartite entanglement measures are universal quantities that depend only on the emergent low-energy theory. For a gapped system, we argue that either $g\neq 0$ and $h=0$ or $g=h=0$, depending on whether the ground state has long-range order. For a critical system, we develop a numerical algorithm for computing $g$ and $h$ from a lattice model.
We compute $g$ and $h$ for various CFTs and show that $h$ depends only on the central charge whereas $g$ depends on the whole operator content. 
\end{abstract}

\maketitle
Quantum entanglement has come to play a key role in our understanding of emergent phenomena in quantum many-body physics and modern numerical methods. 
Most attention has focused  on bipartite entanglement, e.g.\ properties of a pure state  on two parties $\ket{\psi}_{AB}$.
The entanglement entropy $S(A)$ is  the unique measure of bipartite entanglement because, up to reversible local operations and classical communication, the EPR pair is the unique form of bipartite entanglement. 
In contrast, a pure tripartite state $\ket{\psi}_{ABC}$ admits a large (presumably infinite) number of distinct forms of  entanglement, and consequently a variety of tripartite entanglement measures have been proposed \cite{walter2016multipartite}. But it remains relatively unexplored what universal features such measures might reveal about a  many-body system \cite{Audenaert2002,Marcovitch2009,Bayat2010,Bhattacharyya2019,Kudler-Flam2019,Kusuki2019,Bayat2017,Gray2019}.

Recently two tripartite entanglement measures, the entanglement of purification $E_P(A:B)$ \cite{Terhal2002} and the ``reflected entropy’’ $S_R(A:B)$  \cite{Dutta2019} have been applied to many-body physics within the context of holographic duality. 
As motivation, recall that the Ryu-Takayanagi formula equates the bipartite entanglement entropy of a boundary theory to the area of a minimal surface in its holographic dual \cite{Ryu2006}, a central result in the effort to relate the emergence of spacetime geometry to quantum entanglement.
It is then natural ask whether there are multi-partite entanglement measures which might also have a dual geometric interpretation. 
In Refs.~\cite{Takayanagi2017,Nguyen2018} it was conjectured  that the minimal cross section of the bulk ``entanglement wedge’’ joining two parties, $E_W(A:B)$, is dual to the entanglement of purification in the boundary, $E_P = E_W$. 
More recently, however, by developing a field-theoretic method for calculating $S_R$ in generic conformal field theories (CFTs), it was shown that $S_R = 2 E_W$ \cite{Dutta2019}.
In general $S_R \neq 2 E_P$, so one possible resolution is that their equality is a special property of holographic CFTs which is violated at subleading order in large-$N$ expansion~\footnote{It has also been argued that the logarithmic negativity is dual to $E_W$~\cite{Kusuki2019}, but $\mathcal{E}_N$ is not lower bounded by $I$, so we do not consider it here.}.
The gap between them, $2 E_P - S_R$, would then constitute an interesting entanglement measure of this violation. But investigating this discrepancy requires a method for computing these quantities in generic many-body systems.

\begin{figure}
    \centering
    \includegraphics[width=0.28\linewidth]{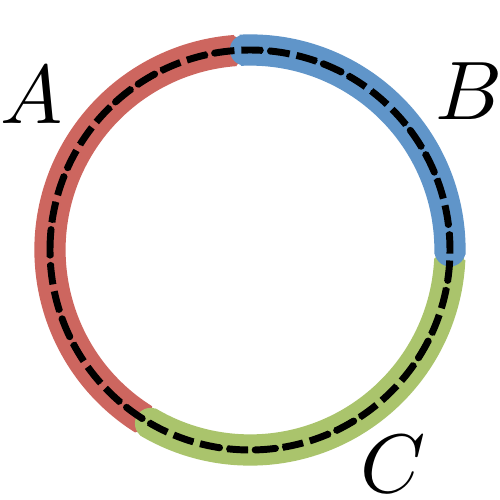}
    \hspace{0.6cm}
    \includegraphics[width=0.28\linewidth]{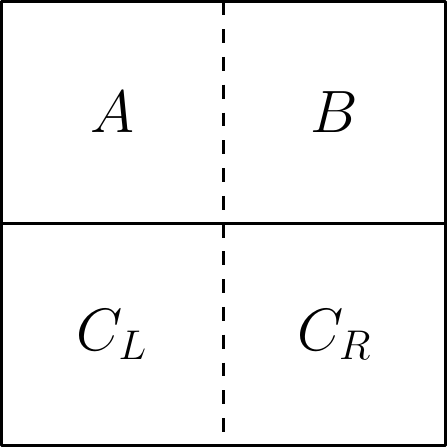}
    \caption{Left: A spin chain on a circle that is divided into three parties $A$, $B$, and $C$.
    Right: Geometry in the computation of $E_P(A:B)$. Region $C$ is divided into $C_L$ and $C_R$. The dashed line represents the entanglement cut between $AC_L$ and $BC_R$.}
    \label{fig:ep_setup}
\end{figure}

In this work we derive a method for computing $E_P$ and $S_R$ in 1D lattice models. To summarize our findings it is convenient to define UV-regularized version of these quantities \footnote{$2E_P(A:B),S_R(A:B)$ and $I(A:B)$ all scale logarithmically with the UV cutoff with the same coefficient $c/3$ in front, see Ref.~\onlinecite{Takayanagi2017,Nguyen2018,Dutta2019}.},  $g(A:B) \equiv 2 E_P(A:B) - I(A:B) \geq 0 $ and $h(A:B) \equiv S_R(A:B) - I(A:B) \geq 0$, where $I$ is the mutual information \footnote{Note that while the constituents are, $g, h$ are not monotonic under quantum operations on $A, B$.}.
For the tripartition of a ring shown in Fig.~\ref{fig:ep_setup}, holographic duality predicts that they take on the universal value $g = h = \frac{c}{3} \log(2)$, where $c$ is the central charge of the CFT%
   ~\footnote{In this computation, the regions $A, B$ are taken to touch and the UV divergences in $E_W, I$ are regulated by  a radial cutoff which is taken to zero after the subtraction $2 E_W - I$~\cite{Nguyen2018}. This prescription corresponds to the lattice regularization employed in our numerical results. An alternative procedure, in which the quantities are regularized by a small spacing between $A, B$, yields a different result \cite{Takayanagi2017}.}.
But what about in a generic lattice model?
As a limiting case, we start by proving structure theorems for states with $g, h = 0$ which imply that $h = 0$ if an only if a state is gapped ($c=0$), while $g = 0$ if and only if the system is gapped \emph{and} does not spontaneously break a symmetry.
We then develop a method for numerically computing $g, h$ from a lattice Hamiltonian on systems up to $N \sim 100$ sites. 
As expected, we find that $h = \frac{c}{3} \log 2$ is universal. However we find that $g \geq h$ and depends on the operator content of the CFT in addition $c$, yet is nevertheless completely universal. 
Thus $2 E_P - S_R = g - h$ constitutes a new and universal tripartite entanglement invariant of CFTs. 

\textit{$E_P$ and $S_R$} --- We first review the definitions of the entanglement of purification $E_P(A:B)$ and reflected entropy $S_R(A:B)$. Unlike the bipartite entanglement entropy, which is a function of a reduced density matrix on one party, these mixed state entanglement measures are functions of the reduced density matrix on two parties, $\rho_{AB}$, or equivalently its purification  $\ket{\psi}_{ABC}$, where $\rho_{AB}=\Tr_{C}\ket{\psi}\bra{\psi}$. 

The entanglement of purification $E_P(A:B)$ \cite{Terhal2002} is the minimum of the entanglement entropy $S_{AC_L}$ over all purifications $\ket{\phi}_{ABC_LC_R}$ of $\rho_{AB}$ to another pair of systems $C = C_L C_R$:
\begin{equation}
\label{eq:Ep_def}
    E_P(A:B) \equiv \min_{\ket{\phi}} S_{AC_L}\big( \ket{\phi}_{ABC_LC_R} \big).
\end{equation}
The partitions of the subsystems are depicted schematically in Fig.~\ref{fig:ep_setup}.
In principle the auxiliary space $C_LC_R$  can be arbitrary, but the minimal $S_{AC_L}$ can always be achieved with $\dim(\calH_{C_L}) ,\, \dim(\calH_{C_R}) \leq \rank(\rho_{AB})$.\cite{Ibinson2008} We may alternatively rephrase Eq.~\eqref{eq:Ep_def}  as a minimization over unitary operations $U_C$ restricted to $C_L C_R$ starting from an \emph{arbitrary} purification $\ket{\phi_0}_{ABC_LC_R}$ of sufficiently large dimension,
\begin{equation}
    E_P(A:B) = \min_{U_{C_L C_R}} S_{AC_L}\big( U_C\ket{\phi_0}_{ABC_LC_R} \big),
\end{equation}
which is the viewpoint taken in our numerical approach.

 $E_P$ is lower bounded by the mutual information~\cite{Terhal2002}, $E_P(A:B)\geq I(A:B) /2$, so we define a non-negative quantity 
\begin{equation}
\label{eq:g_def}
    g(A:B) \equiv 2 E_P(A:B)- I(A:B)\geq 0.
\end{equation}
The physical intuition behind this new quantity is that the subtraction of the mutual information removes correlations which are purely bipartite, as will be made more precise by the structure theorems below.

To define the reflected entropy $S_R(A:B)$, we instead pick a particular purification of $\rho_{AB}$ known as the canonical purification $\ket{\sqrt{\rho_{AB}}}$. It is defined as follows: we first take the unique non-negative square root of the reduced density matrix $\rho_{AB}$, and then regard the operator $\sqrt{\rho_{AB}}$ as a state $\ket{\sqrt{\rho_{AB}}} \in \calH_A\otimes \calH_B \otimes \calH^{*}_A \otimes \calH^{*}_{B}$. The reflected entropy $S_R(A:B)$ is defined as
\begin{equation}
\label{eq:SRdef}
    S_R(A:B)\equiv S_{AA^{*}}\big(\ket{\sqrt{\rho_{AB}}}\big).
\end{equation}
It is shown in Ref.~\onlinecite{Dutta2019} that $S_R(A:B) \geq I(A:B)$, so we define the nonnegative quantity
\begin{equation}
\label{eq:hdef}
    h(A:B) \equiv S_R(A:B)-I(A:B) \geq 0.
\end{equation}

In order to interpret the nature of the tripartite entanglement captured by these quantities, we derive ``structure theorems'' for states which saturate these lower bounds, i.e., states with $g=0$ or $h=0$.

\textit{States with $g(A:B)=0$} --- We first define a class of pure tripartite wavefunctions known as \emph{triangle states}.
\begin{definition}[Triangle State]
A state $\ket{\psi}_{ABC}$ is a triangle state if for each local Hilbert space there exists a bipartition $\calH_\alpha = \calH_{\alpha_L}\otimes \calH_{\alpha_R}$ ($\alpha=A,B,C$) such that 
\begin{align}
\ket{\psi}_{ABC}= \ket{\psi}_{A_RB_L}\ket{\psi}_{B_RC_L}\ket{\psi}_{C_RA_L} ,
    \label{eq:triangle}
\end{align}
where $\ket{\psi}_{\alpha_R\beta_L}$ are pure states in $\calH_{\alpha_R}\otimes\calH_{\beta_L}$.
\end{definition}
In other words, a triangle state can be obtained by pair-wise distributing bipartite-entangled states followed by local unitaries. In this sense, a triangle state lacks nontrivial tripartite entanglement. We prove the following theorem in the Supplemental Material (SM)~\cite{Supp, Haah2016}.
 
\begin{theorem}\label{thm:zero_g}
A state $\ket{\psi}_{ABC}$ is a triangle state up to local isometries if and only if $g(A:B)=0$.
\end{theorem}
The ``only if" direction can be shown by noting that $2E_P(A:B) = I(A:B)$ in the purification $\ket{\psi}_{ABC}$ of $\rho_{AB}$.
The proof of the ``if" direction is more complicated, and is presented in SM~\cite{Supp}. 

Conversely, $g(A:B) > 0$ implies that $\ket{\psi}_{ABC}$ contains tripartite entanglement that cannot be factorized pairwise. 
For example, for a GHZ state $\ket{\psi}_{ABC}=\sqrt{d^{-1}}\sum_{j=1}^d \ket{j_Aj_Bj_C}$ the optimal purification of $\rho_{AB}$ is $\ket{\psi}_{ABC}$ itself \cite{Nguyen2018}, resulting in $g(A:B)=\log d$. It can also be shown that the $W$~state has nonzero $g(A:B)$. This is in accordance with the fact the GHZ state and $W$~state are not triangle states \cite{Dur2000}.

\textit{States with $h(A:B)=0$} ---  
It can be verified that a triangle state has $h(A:B)=0$, so  $h(A:B)\neq 0$ also implies irreducible tripartite entanglement.  But for the GHZ state,  $g(A:B)\neq 0$ while $h(A:B)=0$, which suggests that that some forms of tripartite entanglement are ``invisible'' to $h$.

To make this precise we introduce the notion of sum of triangle states.
\begin{definition}[sum of triangle states (SOTS)]
A pure state $\ket{\psi}_{ABC}$ is a SOTS if for each local Hilbert space $\calH_\alpha$ there exists a decomposition $\calH_\alpha = \bigoplus_j \calH_{\alpha^{j}_L}\otimes \calH_{\alpha^{j}_R}$ such that
\begin{equation}
\label{eq:SOTS}
    \ket{\psi}_{ABC}=\sum_{j} \sqrt{p_j} \ket{\psi_{j}}_{A^{j}_RB^{j}_L}\ket{\psi_{j}}_{B^{j}_RC^{j}_L}\ket{\psi_{j}}_{C^{j}_RA^{j}_L},
\end{equation}
where $\ket{\psi_{j}}_{\alpha^{j}_R \beta^{j}_L}$ represents a pure state in $\calH_{\alpha^{j}_R}\otimes \calH_{\beta^{j}_L}$, etc, and $\sum_{j} p_j =1$.
\end{definition}
For example, the GHZ state is a SOTS with $p_j = \frac{1}{d}$ and the triangle state is a SOTS for which $p_j=1$  for exactly one $j$. 
By using the structure theorem for states satisfying strong subadditivity \cite{Hayden2004}, we prove~\cite{Supp} the following:
\begin{theorem} A state $\ket{\psi}_{ABC}$ is a SOTS if and only if  $h(A:B)=0$.
\label{thm:zero_h}
\end{theorem}
As a corollary, while in general $h(A:B)\neq h(B:C)\neq h(C:A)$, if one vanishes then all of them vanish (and likewise for $g$).

\textit{$g$ and $h$ for 1D gapped systems} ---
We now give a physical interpretation of these structure theorems in the context of 1D Hamiltonians: we argue that on a ring with the tripartition shown in Fig.~\ref{fig:ep_setup}, a system is gapped if and only if $h = 0$, and gapped without long-range order if and only if $g = 0$.
As motivation, consider the two limiting gapped phases of the 1D Ising model: the symmetric paramagnet, $\ket{\mathrm{PM}} = \ket{\rightarrow \rightarrow \cdots}$, and the ferromagnet $\ket{\mathrm{FM}} = \frac{1}{\sqrt{2}}\big( \ket{\mathord{\uparrow\uparrow}\cdots} + \ket{\mathord{\downarrow\downarrow}\cdots} \big)$.
When partitioned into 3 subsystems, the $\ket{\mathrm{PM}}$ ($\ket{\mathrm{FM}}$) state corresponds to
a product state  (GHZ state), so it will have $g = 0$ ($g=\log2$) and for both, $h=0$. Indeed, we see that $g$ is sensitive to the ``cat state'' structure of the exact ground state in a symmetry-broken phase, so will generically detect the multiplicity of super-selection sectors. 
Away from these extremal points, the ground state develops additional short-range entanglement. 
However, so long as sizes of the regions $N_A, N_B, N_C$ are larger than the correlation length $\xi$, this additional entanglement simply dresses the product state within each superselection sector into a triangle state, and so with exponential accuracy in $N / \xi$, $g$ and $h$ are unchanged.

The argument can be phrased most precisely in the language of matrix product states.
%as shown in Refs.~\onlinecite{Verstraete2006, Arad2013}, area-law ground states can be faithfully represented by finite-dimensional MPS.
We first take a finite-dimensional MPS as an approximation to the ground state of a 1D system \footnote{There is a caveat to use the finite-dimensional MPS as an approximation. It is only rigorously proven that the state can be faithfully represented if bond dimension grows with the system size~\cite{Verstraete2006, Arad2013}. For a finite bond dimension, it has only been shown that local properties can be well approximated \cite{Dalzell2019}. In order to make our argument, we have to assume that the finite bond dimension does not result in a substantial error in $g$ or $h$, which are nonlocal properties of the ground state. Despite not rigorous proven, the assumption is highly plausible because of empirical success of infinite MPS algorithms, where correlation functions and entanglement properties at long distances are extracted from a finite-dimensional MPS. Therefore the argument could be regarded as heuristic for a general gapped theory. However, it is rigorous if the ground state can be exactly represented by a finite-dimensional MPS, for example that of a MPS parent Hamiltonian \cite{Perez-Garcia2007}.}.
The thermodynamic limit is taken by fixing $N_A/N,N_B/N$ and taking $N\rightarrow\infty$, where $N$ is the total system size. In the thermodynamic limit we can then apply the standard MPS coarse-graining procedure \cite{Verstraete200502} to obtain a fixed-point MPS.
If the initial correlation length is finite \cite{Hastings2006}, the state flows to an MPS with $\xi = 0$. It is straightforward to show that a $\xi = 0$ MPS is precisely the $N$-party generalization \footnote{The generalization of a triangle state to many parties is a polygon state, which is discussed in detail in \cite{Supp}. A polygon state is a triangle state with respect to any tripartition into contiguous regions.} of a triangle state \cite{Verstraete200502, Chen2011}, so by the structure theorems we obtain $g = h = 0$.
On the other hand, if the MPS has an infinite correlation length (e.g., it is a cat state as occurs for spontaneous symmetry breaking or phase coexistence), then it flows to a \emph{sum} of $\xi = 0$ MPS which are locally orthogonal \cite{Perez-Garcia2007, Supp}. Thus in the long-range ordered phase we have $g\neq 0$ and $h=0$. These cases are analyzed in greater detail in \cite{Supp}. Note  that the precise statement of our claim is thus as follows:  A fixed-point MPS has $h(A:B) = 0$ for all contiguous tripartitions.  Since all MPS flow towards fixed-point MPS under coarse graining,  $h(A:B) \to 0$ as $N_A, N_B \to \infty$ \footnote{Technically, taking the limit assumes the continuity of $g(A:B)$ and $h(A:B)$ with respect to the reduced density matrix $\rho_{AB}$ \cite{Supp}. The continuity property of $E_P$ and $S_R$ has already been proven in \cite{Terhal2002, Akers2020}.}.

\textit{Gapless systems} ---
At a critical point $g$ and $h$ need not vanish. In fact, they are universal constants which depend only on the emergent CFT in the thermodynamic limit.

We now briefly describe the algorithm to compute $g$ and $h$ of the ground state of a critical quantum spin chain with $N$ sites and Hamiltonian $H$.
First the ground state $\ket{\psi}_{ABC}$ is obtained in the form of a periodic uniform MPS (puMPS) \cite{Zou2018,Zou202001,Zou202002}. A puMPS consists of $N$ copies of the same rank-3 tensor $M$ with dimensions $D\times D\times d$, where $d$ is the dimension of the Hilbert space on each site, and $D$ is the bond dimension which grows polynomially with the system size $N$ (Fig.~\ref{fig:puMPS}).
The tensor $M$ is obtained variationally by minimizing the expectation value of $H$. We then apply the standard MPS coarse-graining procedure \cite{Verstraete200502,Supp} to ``compress'' the Hilbert space of each region down to a smaller one via a sequence of isometries,
$\calH_{\alpha} \to \calH_{\tilde{\alpha}}$.
Because the entropy of each region is sub-extensive, $S_{\alpha} \ll N_\alpha \log(d)$ -- even at a critical point -- we can reduce the dimension of the Hilbert space $\tilde{d}_{\alpha} \ll d_{\alpha}$
while preserving all the bipartite and tripartite entanglement properties among the three parties $A$, $B$ and $C$ to high-accuracy.
The coarse-grained state $\ket{\tilde{\psi}}_{\tilde{A}\tilde{B}\tilde{C}}$ can be represented by a MPS with three tensors $M_{\alpha}$ with dimensions $D\times D\times \tilde{d}_{\alpha}$ (Fig. 2), where $\tilde{d}_{\alpha}\leq D^2$. 
$D, \tilde{d}_{\alpha}$ grow polynomial with system size; as an example, for the Ising CFT we use $D = 12$, $\tilde{d}_{\alpha} = 36$ for $N = 24$ and $D = 26$, $\tilde{d}_{\alpha} = 100$ for $N = 84$.

\begin{figure}
    \centering
    \includegraphics[width=0.8\linewidth]{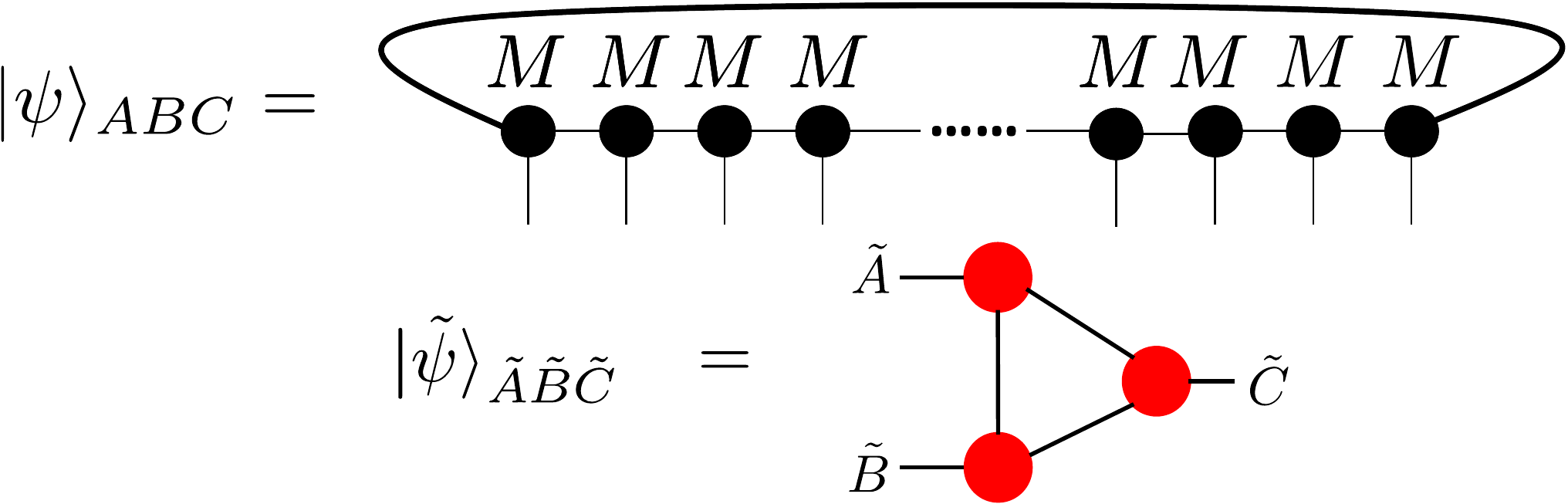}
    \caption{The state before and after coarse-graining. Top: The periodic uniform matrix product state (puMPS) represents the ground state of a translation-invariant critical quantum spin chain before coarse-graining. Bottom: The puMPS is coarse-grained into a MPS with 3 tensors corresponding to the coarse-grained Hilbert spaces $\calH_{\tilde{A}},\calH_{\tilde{B}},\calH_{\tilde{C}}$.}
    \label{fig:puMPS}
\end{figure}
We compute $S_R(A:B)$ according to Eq.~\eqref{eq:SRdef} in the dense representation.
Assuming that $\tilde{d}_{A}\leq\tilde{d}_{B}$, the total time cost scales as $\mathcal{O}\big( \tilde{d}^4_A \tilde{d}^2_B \big)$. To compute $E_P(A:B)$, we first make a random split of $\calH_{\tilde{C}}$ into $\calH_{\tilde{C}_L}\otimes \calH_{\tilde{C}_R}$ with dimensions $\tilde{d}_{C_L}\times\tilde{d}_{C_R}$. We then numerically minimize the entanglement entropy of $\tilde{A}\tilde{C}_L$ with respect to a unitary $U_{\tilde{C}}$ on $\tilde{C}$,
\begin{align}
    E_P(A:B) = \min_{U_{\tilde{C}}} S_{\tilde{A}\tilde{C}_L} \Big(U_{\tilde{C}}\ket{\tilde{\psi}}_{\tilde{A}\tilde{B}\tilde{C}}\Big).
\end{align}
We verified numerically that the $\tilde{d}_\alpha$ are large enough to achieve the (near) optimal purification. The numerical optimization can be performed with, e.g.,  the non-linear conjugate gradient algorithm, since the gradient can be constructed explicitly (see \cite{Supp}). The time cost of each gradient calculation scales as $\mathcal{O}(\tilde{d}^{2}_A \tilde{d}_B \tilde{d}^{2}_C)$, assuming that $\tilde{d}_{A}\leq\tilde{d}_{B}$. The mutual information $I(A:B)$ can also be computed using the coarse-grained state, with time cost $O(\tilde{d}^3_{\max})$, where $\tilde{d}_{\max}\equiv \max_\alpha\{\tilde{d}_{\alpha}\}$. 

\textit{$g$ and $h$ for various CFTs} --- In order to show that $g$ and $h$ are universal, we study the Ising model with an irrelevant near-to-nearest neighbor interaction~\cite{Obrien_2018},
\begin{align}
    \label{eq:Ising_TCI}
    H=\sum_{j=1}^N \left[\begin{array}{l}
        -X_j X_{j+1}-Z_j \\ {\;}\mathop{+}\lambda\big(X_j X_{j+1} Z_{j+2} + Z_j X_{j+1} X_{j+2}\big) \end{array}\right],
\end{align}
where $X_j$($Z_j$) are Pauli $X$($Z$) matrices on sites $j$ and periodic boundary conditions are assumed. The model is critical described by the Ising CFT for $\lambda<\lambda^*$, 
    gapped for $\lambda>\lambda^*$, where the transition at $\lambda^{*}\approx 0.428$ is described by the tricritical Ising CFT \cite{Obrien_2018}.
We study four parameter values, $\lambda=0, 0.3, 0.4, \lambda^{*}$, where the first three correspond to the Ising CFT and the last correspond to the tricritical Ising CFT. 

We fix $N_A=N_B=N_C=N/3$ and compute $g(A:B)$ and $h(A:B)$ as a function of $N$, shown in Fig.~\ref{fig:Ising_g_h}. We see that both $g$ and $h$ converge to a constant as $N\rightarrow\infty$ \footnote{The extrapolation is based on the numerical observation that $g$ and $h$ converges as a power of $1/N$, where the power is determined empirically.}. Furthermore, the constant is the same for $\lambda=0$ and $\lambda=0.3$, indicating that $g$ and $h$ are universal. We denote the universal quantities as $g^{\CFT}$ and $h^{\CFT}$. At $\lambda=\lambda^{*}\approx 0.428$, we obtain a different value that corresponds to the tricritical Ising CFT. At $\lambda=0.4$, both $g$ and $h$ go through a renormalization group flow from the tricritical Ising CFT to the Ising CFT, analogous to the spectral flow observed in Ref.~\onlinecite{Zou2018}. The values of $g^{\CFT}$ and $h^{\CFT}$ for various CFTs are summarized in Table 1.

\begin{figure}
    \includegraphics[width=0.9\linewidth]{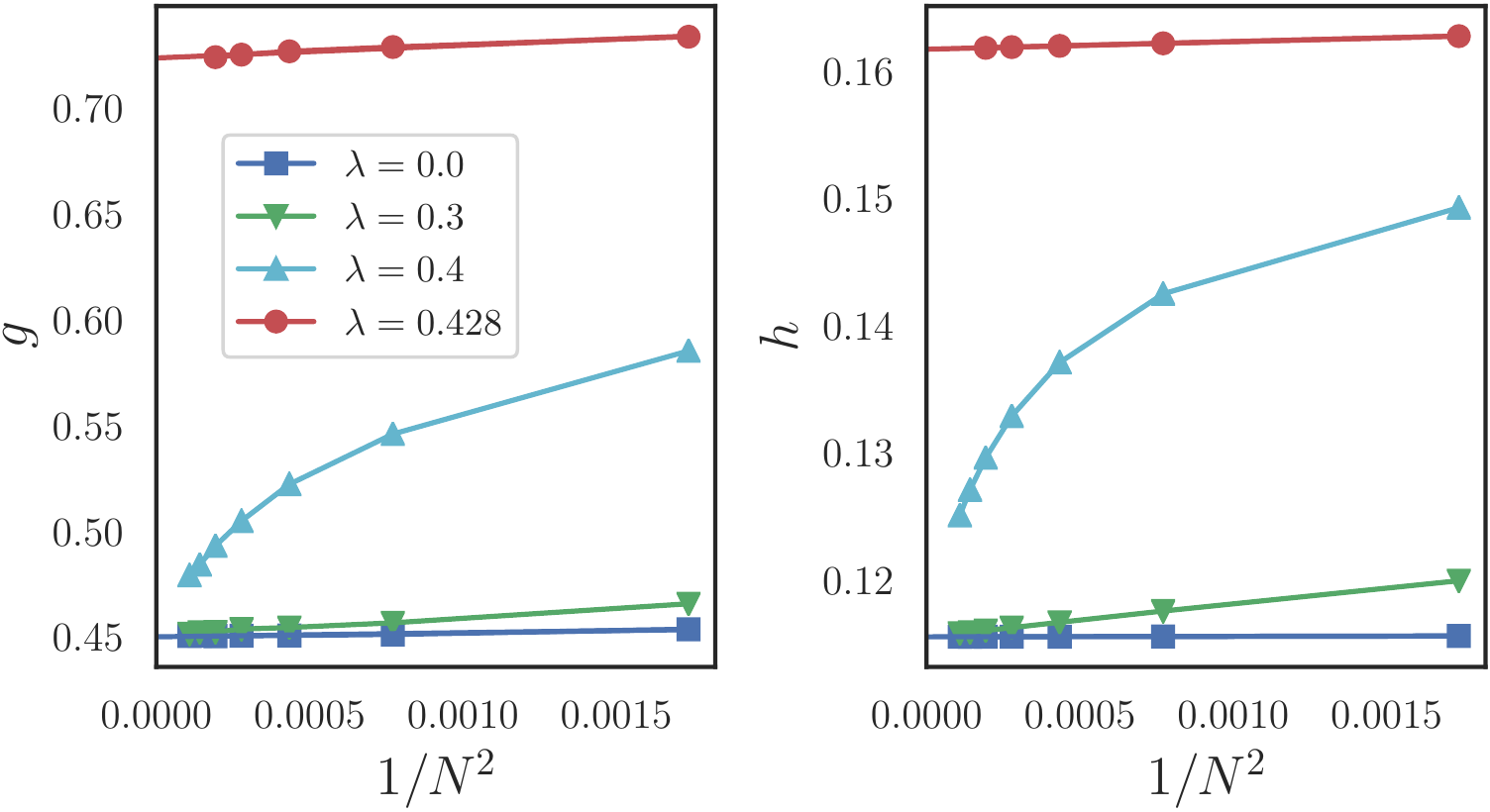}
    \caption{ 
     $g(A:B)$ and $h(A:B)$ from the model Eq.~\eqref{eq:Ising_TCI} with different $\lambda$'s. At $\lambda=0$ and $\lambda=0.3$, the quantities converge to $g^{\CFT}$ and $h^{\CFT}$ of the Ising CFT. At $\lambda=\lambda^{*}\approx 0.428$, both quantities converge to a different value that corresponds to the tricritical Ising CFT. At $\lambda=0.4$ we observe a renormalization group flow from the tricritical Ising CFT to the Ising CFT. 
     }
    \label{fig:Ising_g_h}
\end{figure}

\begin{table}[th]
    \centering
    \begin{tabular}{ |c| *{4}{@{\;\,}c@{\,\;}|} } \hline
    Theory & $c$ & $g^\CFT$ & $h^\CFT$ & $\tfrac{c}{3}\log2$   \\[0.2ex]\hline
    gapped symmetric & 0 & 0 & 0 & 0  \\\hline
    long-range ordered & 0 & $> 0$ & 0 &  0 \\\hline
    Ising CFT             &  $1/2$ & 0.450 & 0.1155 & 0.11553   \\\hline
    Tricritical Ising CFT & $7/10$ & 0.719 & 0.1617 & 0.16173   \\\hline
    Free boson $R=\sqrt{3}$  & $1$ & 0.920 & 0.2310 & 0.23105   \\\hline
    Free boson $R=2$         & $1$ & 0.899 & 0.2310 & 0.23105   \\\hline
    Free boson $R=\sqrt{6}$  & $1$ & 0.906 & 0.2310 & 0.23105   \\\hline
    \end{tabular}
    \caption{$g^\CFT$ and $h^\CFT$ extracted numerically through finite size scaling. For the gapped spin chains, the universal values of $g$ and $h$ are shown. For the gapless spin chains, we show the central charge  $c$ as well as $g^\CFT$ and $h^\CFT$ of the CFTs.}
    \label{table1}
\end{table}
We also verified that the values of $g^{\CFT}$ and $h^{\CFT}$ do not depend on the relative sizes of $A,B,C$~\cite{Supp}.
For any ratio $N_A/N$ and $N_B/N$, once we take the thermodynamic limit $N\rightarrow\infty$, both $g(A:B)$ and $h(A:B)$ converge to the universal constants $g^{\CFT}$ and $h^{\CFT}$.

We proceed to examine how $g^\CFT$ and $h^\CFT$ depend on the conformal data of the CFT. To do so we compute $g^\CFT$ and $h^\CFT$ for the free compactified boson CFT for differing compactification radius $R$. All have the same central charge $c=1$, but the operator content depends on  $R$. A concrete lattice realization of the CFT is the XXZ model,
\begin{align}
\label{eq:XXZ}
    H=\sum_{j} \big(X_j X_{j+1}+Y_j Y_{j+1} + \Delta Z_j Z_{j+1}\big),
\end{align}
where $-1 \leq \Delta < 1$ is a parameter that determines the compactification radius $R=\sqrt{2\pi/\cos^{-1}(-\Delta)}$.
We compute $g^{\CFT}$ and $h^{\CFT}$ for different $R$'s by extrapolating $g(A:B)$ and $h(A:B)$ for different $\Delta$'s to the thermodynamic limit. The result is shown in Fig.~\ref{fig:XXZgs_hs} and Tab.~\ref{table1}, where $R=\sqrt{3},2,\sqrt{6}$ correspond to $\Delta=0.5,0,-0.5$, respectively \footnote{Note that at $\Delta=0$ both $g^{\CFT}$ and $h^{\CFT}$ equal twice of that for the Ising CFT, in accordance with the duality between the CFTs.}. 

We see that $h^{\CFT}$ does not depend on $\Delta$ and is compatible with $h^\CFT=\frac{c}{3}\log 2$.
On the other hand, $g^{\CFT}$ depends on $\Delta$ and thus on $R$. For example, as shown in Table I, $g^\CFT$ takes on three different values at $\Delta=0,0.5,-0.5$, which correspond to $R=2,\sqrt{3},\sqrt{6}$, respectively.  We conclude that $h^{\CFT}$ only depends on the central charge but $g^{\CFT}$ depends on the whole operator content. This feature of $h^{\CFT}$ can be understood as follows. The canonical purification of $\rho_{AB}$ can be regarded as the ground state of a CFT living on a circle, divided into four contiguous segments $A,B,\bar{B},\bar{A}$.  The measure $h(A:B) = S_{A\bar{A}} - S_A - S_B + S_{AB}$ involves only contiguous pieces and is hence proportional to the central charge.

\begin{figure}
    \centering
    \includegraphics[width=0.9\linewidth]{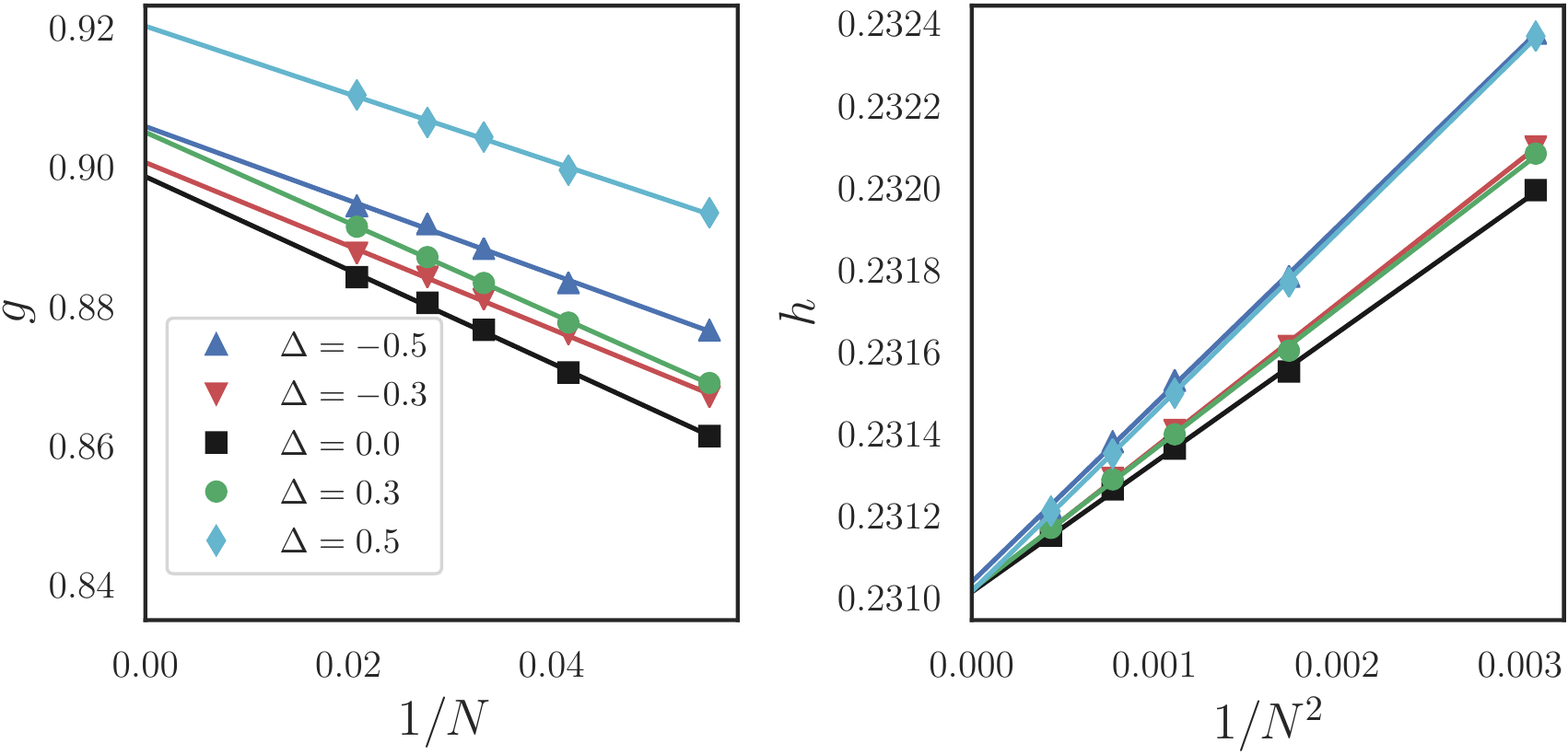}
    \caption{$g(A:B)$ and $h(A:B)$ from the XXZ model with different $\Delta$'s at sizes $18\leq N \leq 48$. We see that $g^\CFT$ depends on $\Delta$ while $h^\CFT$ is independent of $\Delta$.}
    \label{fig:XXZgs_hs}
\end{figure}

\textit{Discussion} --- 
In this work we have introduced two positive quantities $g$ and $h$ which quantify the obstruction to factorizing a tripartite state into pairwise correlations. 
While the entanglement wedge cross section duality   $E_W = E_P = S_R/2$  predicts $h = g = \frac{c}{3}\log(2)$, for low-$c$ CFTs like the Ising model we find $g > h = \frac{c}{3}\log(2)$. The gap $g - h$ is universal, but it remains an open question how to compute it from the underlying data of the CFT.
It is natural to conjecture a general bound $g \geq h$, which would follow from the monotonicity of $S_R$ under a partial trace.

\begin{acknowledgments}
The authors are grateful to Ning Bao, Jeongwan Haah, and Guifre Vidal for enlightening conversations.
We are particularly indebted to Brian Swingle both for our earlier collaborations and the suggestion to compute $S_R$. 
YZ acknowledges Compute Canada. Research at Perimeter Institute is supported by the Government of Canada through the Department of Innovation, Science and Economic Development Canada and by the Province of Ontario through the Ministry of Research, Innovation and Science. Sandbox is a team within the Alphabet family of companies, which includes Google, Verily, Waymo, X, and others.
KS acknowledges support from the NSF Graduate Research Fellowship Program (Grant No.\ DGE 1752814).
RM is supported by the National Science Foundation under DMR-1848336.
MPZ was supported the DOE, Office of Science, Office of High Energy Physics under QuantISED Award DE-SC0019380 and under contract DE-AC02-05CH11231.
\end{acknowledgments}

\bibliography{main}

\appendix
\onecolumngrid
\section{Vanishing \texorpdfstring{$g$}{g} and triangle states}\label{app:proof_g}
In this section we prove a structure theorem for tripartite quantum states $\ket{\psi}_{ABC} \in \calH_A\otimes\calH_B\otimes\calH_C$ with vanishing $g(A:B)$. We first remind the reader of the definitions of conditional mutual information.
%\begin{definition}[Mutual information]
%Given a density matrix $\rho_{AB}$ on $\calH_A\otimes\calH_B$ and %the reduced density matrices $\rho_A=\Tr_B \rho_{AB}$, $\rho_B=\Tr_A %\rho_{AB}$, the mutual information $I(A:B)=S_A+S_B-S_{AB}$, where %%$S_{\alpha}=-\Tr(\rho_{\alpha}\log \rho_{\alpha})$ is the %von~Neumann entropy.
%\end{definition}
%\noindent Observe that if $\rho_{AB}$ is a pure state, we have $I(A:B)=2S_A=2S_B$.
\begin{definition}[Conditional mutual information]
\label{def:cmi}
Given three parties $A$, $B$ and $C$, the conditional mutual information is defined by $I(A:C|B) \equiv I(A:BC)-I(A:B) = S(AB)+S(BC)-S(ABC)-S(B)$.
\end{definition}

\noindent We will use the following properties of the conditional mutual information:
\begin{enumerate}
    \item $I(A:C|B)=I(C:A|B)$ ~(commutativity);
    \item $I(A:C|B)\geq 0$ ~(strong subadditivity).
\end{enumerate}

\noindent We first re-state the definition of a ``triangle state".
\begin{definition}[Triangle state]\label{def:triangle_state} A pure tripartite state $\ket{\psi}_{ABC} \in \calH_A\otimes\calH_B\otimes\calH_C $ is a triangle state if for each local Hilbert space there exists a bipartition $\calH_\alpha =(\calH_{\alpha_L}\otimes \calH_{\alpha_R}) \oplus \calH_{\alpha}^0$ ($\alpha =A,B,C$) such that
\begin{equation}
	\ket{\psi}_{ABC}= \ket{\psi}_{A_RB_L}\ket{\psi}_{B_RC_L}\ket{\psi}_{C_RA_L} .
	\label{eq:triangle_app}
\end{equation}
\end{definition}
\noindent Observe that $\ket{\psi}_{ABC}$ has no support in any $\calH_{\alpha}^0$.
Note that for notational clarity, the main text version of this definition defines the bipartition as $\calH_\alpha =\calH_{\alpha_L}\otimes \calH_{\alpha_R}$, and as a result that equivalence is only up to local isometry. 
The triangle state is a tensor product of pure states which may be entangled between at most two out of three parties.

\vspace{3ex}
We now state and prove the first structure theorem~\cite{Haah2016}:
\begin{theorem}[States with vanishing $g$]\label{thm:zero_g_app}
A pure tripartite quantum state $\ket{\psi}_{ABC}$ is a triangle state if and only if for any two-party reduced density matrix ($\rho_{AB}, \rho_{BC},\rho_{CA}$), $g(A:B) = g(B:C) = g(C:A) = 0$, respectively.
\end{theorem}

\noindent To prove Theorem~\ref{thm:zero_g_app}, we will need the following theorem from Hayden et al.~\cite{Hayden2004}:
\begin{theorem}[Quantum Markov property \cite{Hayden2004}]\label{thm:quantum_markov_chain}
Let $\rho_{ABC}$ be a quantum state on $\calH=\calH_A \otimes \calH_B \otimes \calH_C$. Then $\rho_{ABC}$ satisfies strong subadditivity with equality---i.e., $I(A:C|B) = 0$---if and only if there exists a decomposition of $\calH_B$
\begin{equation}
    \calH_B = \bigoplus_{i} \calH_{B_L}^i \otimes \calH_{B_R}^i
\end{equation}
such that
\begin{align}
    \rho_{ABC} = \sum_i q_i \, \rho_{AB_L}^i \otimes \rho_{B_RC}^i \,,
    \label{eq:hayden_decomp}
\end{align} 
where $\rho_{AB_L}^i , \rho_{B_RC}^i$ are density matrices in $\calH_A \otimes \calH_{B_L}^i , \calH_{B_R}^i \otimes \calH_C$ respectively,
	and $\{q_i\}$ forms a probability distribution $[$i.e., $q_i\geq 0$ and $\sum_{i}q_i=1$$]$.
\end{theorem}
\noindent For a pure tripartite state $\ket{\psi}_{ABC}$ satisfying $I(A:C|B)=0$, since the density matrix of a pure state is not convex combinations of mixed states, the sum contains only one term and $\ket{\psi}_{ABC} = \ket{\psi}_{AB_L}\ket{\psi}_{B_RC}$.

\begin{proof}[Proof of Theorem~\ref{thm:zero_g_app}]
One direction of the theorem is straightforward to prove. That is, if a pure quantum state $\rho_{ABC}$ is a triangle state, then for any of the three bipartite reduced density matrices, $g(A:B) = g(B:C) = g(C:A) = 0$. This can be seen by computing, for example, $S(AC_L)$ and $I(A:B)$ for Eq.~\eqref{eq:triangle_app}, and observing that $S(AC_L) = I(A:B)/2$. Since any purification provides an upper bound on $E_P(A:B)$, this must be the optimal one, as the lower bound is saturated. Similar computation may be done for $g(B:C)$ and $g(C:A)$ for reduced density matrices of the triangle state to show that these, too, must equal zero.

To prove the converse, consider an optimal purification of $\rho_{AB}$ to $C=C_L\otimes C_R$. Then we can write $g(A:B)$ as a sum of two non-negative quantities, in two different ways
\begin{subequations} \label{eq:g_decompositions} \begin{align}
    g(A:B) &= I(C_L:BC_R|A) + I(C_R:A|B) \label{eq:g_decomposition1} \\ &= I(C_R:AC_L|B) + I(C_L:B|A) . \label{eq:g_decomposition2}
\end{align} \end{subequations}
The decompositions follow from rewriting $g$ as $g(A:B) = 2S(AC_L) - I(A:B) = I(AC_L:BC_R) - I(A:B)$ and repeated use of Definition~\ref{def:cmi}.
Thus, if $g(A:B)=0$, then all four of the condition mutual informations vanishes:
\begin{align}
    I(C_L:B|A) = I(C_R:A|B) = I(C_L:BC_R|A) = I(C_R:AC_L|B) = 0 .
    \label{eq:zero_cmi}
\end{align}
%We show this by rewriting Eq.~\eqref{eq:g0} in terms of conditional mutual informations:
%where we have used monogamy equality for pure states.
%where we have repeatedly use Definition~\ref{def:cmi} and the commutativity property of conditional mutual information. By strong subadditivity, each term above must individually be equal to zero and we thus prove Eqs.~\eqref{eq:zero_cmi}.
In quantum information language Eqs.~\eqref{eq:zero_cmi} imply that the state $\rho_{ABC_LC_R}$ forms a quantum Markov chain across the tripartitions $C_L|A|BC_R$ and $AC_L|B|C_R$. 

Applying Theorem~\ref{thm:quantum_markov_chain} to Eq.~\eqref{eq:g_decompositions}, we obtain a decomposition along the conditioned Hilbert space of the optimally purified state. However, for either tripartition $C_L:A:BC_R$ or $AC_L:B:C_R$, the state in question is pure, so as observed earlier, the sum in Eq.~\eqref{eq:hayden_decomp} can contain only one term. The pure state must therefore be a tensor product of two pure states. Decomposing $\calH_A=\calH_{A_R}\otimes \calH_{A_L}$ for the quantum Markov chain across $C_L:A:BC_R$, we obtain
\begin{equation}
    \ket{\psi}_{ABC_LC_R} = \ket{\psi}_{C_LA_R} \ket{\psi}_{A_LBC_R} .
\end{equation}
Next, using the other Markov chain $AC_L:B:C_R$ in Eq.~\eqref{eq:zero_cmi} we may write $ \ket{\psi}_{A_LBC_R}$ as yet another tensor product of pure states: 
\begin{equation}
    \ket{\psi}_{A_LBC_R} = \ket{\psi}_{A_LB_R} \ket{\psi}_{B_LC_R} .
\end{equation}
The optimal purification is therefore decomposed as
\begin{equation}
    \ket{\psi}_{ABC_L C_R} = \ket{\psi}_{C_LA_R} \ket{\psi}_{A_LB_R} \ket{\psi}_{B_LC_R}
\end{equation}
and is therefore a triangle state. The original state is a purification which is equivalent to this triangle state up to a local isometry on $C$. We note that by looking at $g(A:B)$, we deduce that $\ket{\psi}_{ABC}$ is locally isometric to a triangle state. However, the same argument may be applied using any pair of parties.
\end{proof}

\section{Vanishing \texorpdfstring{$h$}{h} and sums of triangle states (SOTS)}\label{app:proof_h}
In this section we define a classes of states called \emph{Sum of polygon states} (SOPS) and discuss their properties.
The main result of this section is the structure theorem for quantum states with vanishing $h(A:B)$ (Theorem~\ref{thm:zero_h_app}).

%First, we introduce some notations and definitions. In what follows, we abuse notation slightly to simplify the manipulations. \RM{Please never use this notation again} We label states by the Hilbert spaces in which they live. For example, $\ket{\calH_1 \calH_2}$ is an \emph{arbitrary} state in $\calH_1\otimes \calH_2$, not necessarily a product state. To indicate quantum states which \emph{are} product states, we explicitly divide the ket denoting the state into two kets, such as $\ket{\calH_1} \ket{\calH_2}$.

\begin{definition}[Splitting]\label{def:splitting}
A splitting of a Hilbert space $\calH_i$ is an orthogonal decomposition of the Hilbert space into a direct sum of tensor product spaces 
\begin{align}
    \calH_i = \calH_i^0 \;\oplus\; \bigoplus_j \calH_{iL}^j \otimes \calH_{iR}^j \,.
\end{align}
\end{definition}
\noindent The space $\calH_i^0$ may be 0-dimensional.

% \begin{definition}[Polygon State]
% A pure $N$-party quantum state $\ket{\psi}\in \calH_{1}\otimes \calH_{2}\dots \otimes \calH_{N}$ if for each $i$, $\calH_{i} = \calH_{i}^L \otimes \calH_{i}^R\oplus \calH_i^{(0)}$.
% \begin{equation}
%     \ket{\psi} = \otimes_{i=1}^N \ket{\phi}_{i,R}\ket{\phi}_{f(i),L}
% \end{equation}
% where $f(i) = (i+1)\mod N +1$ and $\ket{\phi}_{i,R(L)} \in \calH_i^{R(L)}$.
% \end{definition}

\begin{definition}[Sum of polygon states---SOPS]\label{def:SOPS}
An $N$-party pure quantum state $\ket{\psi} \in \calH_1 \otimes \calH_2 \otimes \cdots \otimes \calH_N$ is a SOPS with respect to the decomposition $(\calH_1,\calH_2,\dots,\calH_N)$ if for each party $i$, $\calH_i$ admits a splitting and
\begin{align}
	\ket{\psi} &= \sum_j a_j \bigotimes_i \ket{j}_{(iR)(i^+L)}
	\qquad \mathrm{such\ that\ } \ket{j}_{(iR)(i^+L)} \in \calH_{iR}^j \otimes \calH_{i^+L}^j \,.
	\label{eq:SOPS_definition}
\end{align}
where $i^+ \equiv (i\mod N) +1$ denote the party after $i$, the coefficients are normalized to $\sum_j |a_j|^2 = 1$.
%and $\bigket{\calH_{i,j}^{R}\calH_{i^+,j}^{L}}$ is a pure state in $\calH_{i,j}^{R}\otimes \calH_{i^+,j}^{L}$.
\end{definition}
We may, without loss of generality, take $a_j \in \mathbb{R}$ and $a_j \geq 0$ by absorbing its phase into one of $\ket{j}_{(iR)(i^+L)}$.
A SOTS is a special case of SOPS with $N=3$;
the triangle state defined in Eq.~\eqref{eq:triangle_app} may be seen as a special case of Eq.~\eqref{eq:SOPS_definition} in which $a_j=1$ for exactly one $j$.
The decomposition is invariant under a cyclic permutation, e.g.\ if a state is a SOPS with respect to $(A,B,C)$, it is also a SOPS with respect to $(B,C,A)$.

We now show that a SOPS is still a SOPS when the local Hilbert spaces are re-defined to include their nearest neighbors, a procedure we term ``coarse-graining".
For example, if we define $\calH_{1:3} = \calH_1 \otimes \calH_2 \otimes \calH_3$, then a state which is an SOPS for $\calH_{1}\otimes \calH_{2}\dots \otimes \calH_{N}$ decomposition is also a SOPS for $\calH_{1:3}\otimes \calH_{4} \dots \otimes \calH_{N}$.
Note that an arbitrary combination (such as one which is not nearest-neighbor) does not necessarily preserve the SOPS structure.

\begin{lemma}[SOPS structure preserved under coarse-graining]\label{lemma:SOP_CG}
Let $N\geq3$.  If $\ket{\Psi} \in \calH_1\otimes\cdots\otimes\calH_N$ is an $N$-party SOPS, then $\ket{\Psi}$ is an $(N{-}1)$-party SOPS with respect to the decomposition $(\calH_{1},\dots,\calH_{i:i+1},\calH_{i+2},\dots,\calH_N)$, where $\calH_{i:i+1} = \calH_i \otimes \calH_{i+1}$.
\begin{proof}
Without loss of generality, we show this for $i=1$; we coarse-grain $\calH_1\otimes \calH_2 \to \calH_{1:2}$.
Because 1-2 always appear in Eq.~\eqref{eq:SOPS_definition} through the combination $\ket{j}_{(NR)(1L)} \ket{j}_{(1R)(2L)} \ket{j}_{(2R)(3L)}$,
    we can identify the splitting $\calH_{1:2,L}^j = \calH_{1L}^j \otimes \calH_{1R}^j \otimes \calH_{2L}^j$ and $\calH_{1:2,R}^j = \calH_{2R}^j$.
Identifying $\ket{j}_{(NR)(1:2,L)} = \ket{j}_{(NR)(1L)} \ket{j}_{(1R)(2L)}$, the claim follows. 
\end{proof}
\end{lemma}

SOPS satisfy a number of interesting properties which we now state.
\begin{porism} \label{lemma:SOPS_cyclicMarkov}
Let $N\geq4$ and $\ket{\psi}$ be a SOPS with respect to $(\calH_1,\dots,\calH_N)$.  Then $I(\calH_{i-1},\calH_{i+1}|\calH_i) = 0$ for all $i$.
\end{porism}
To simplify notation, here we identify the parties (and their respective Hilbert spaces) $0 \leftrightarrow N$ and $1 \leftrightarrow N+1$.

\begin{proof}
Without loss of generality we take $i=2$.  Under coarse-graining, it suffices to prove the statement for $N=4$.
For a state $\ket{\psi}_{1234} \in \mathrm{SOPS}(N=4)$ written in the form Eq.~\eqref{eq:SOPS_definition}, the density matrix $\rho_{123}$ is
\begin{align}
	\rho_{123} = \sum_j |a_j|^2 \Big( \rho_{1L}^j \otimes \ket{j}_{(1R)(2L)}\bra{j}_{(1R)(2L)} \Big)
		\otimes \Big( \ket{j}_{(2R)(3L)}\bra{j}_{(2R)(3L)} \otimes \rho_{3R}^j \Big) ,
\end{align}
where $\rho_{1L}^j = \Tr_{4R} \ket{j}_{(4R)(1L)}\bra{j}_{(4R)(1L)}$ is the density matrix within $\calH_{1R}^j$, and $\rho_{3R}^j = \Tr_{4L} \ket{j}_{(3R)(4L)}\bra{j}_{(3R)(4L)}$ the density matrix within $\calH_{3L}^j$.
This is precisely in the form of Eq.~\eqref{eq:hayden_decomp} with $(1,2,3) = (A,B,C)$ of Theorem~\ref{thm:quantum_markov_chain}, and hence $I(\calH_1:\calH_3|\calH_2) = 0$.
\end{proof}

\begin{lemma} \label{lemma:SOPS_canonpurif}
Let $\ket{\psi}$ be a SOPS with respect to $(\calH_1,\dots,\calH_N)$, and $\rho_{a:b}$ be its density matrix for any contiguous subgroup of parties $\calH_a \otimes \calH_{a+1} \otimes \cdots \otimes \calH_b$.
Then the canonical purification of $\rho_{a:b}$ is a SOPS with respect to $(\calH_a,\calH_{a+1},\dots,\calH_{b-1},\calH_b,\calHconj_b,\calHconj_{b-1},\dots,\calHconj_a)$.
\begin{proof}
Again, we may combine all the parties outside of the $a,\dots,b$ range into a single party, which we label as $1$.  Hence, without loss of generality we take $(a,b) = (2,n)$.
The density matrix on $\calH_2\otimes\cdots\otimes\calH_n$ is
\begin{align}
	\rho_{2:n} &= \sum_j |a_j|^2 \, \rho_{2L}^j \otimes \left[ \bigotimes_{i=2}^{n-1} \ket{j}_{(iR)(i^+L)} \bra{j}_{(iR)(i^+L)} \right] \otimes \rho_{nR}^j \,,
	\label{eq:SOPS_part_rho}
\end{align}
where $\rho_{2L}^j = \Tr_{1R} \ket{j}_{(1R)(2L)}\bra{j}_{(1R)(2L)}$ and $\rho_{nR}^j = \Tr_{1L} \ket{j}_{(nR)(1L)}\bra{j}_{(nR)(1L)}$ are the density matrices in their respective spaces after tracing out $\calH_1$.
As each term in Eq.~\eqref{eq:SOPS_part_rho} are orthogonal, its canonical purification is
\begin{align} \begin{aligned}
	\operatorname{CanonPur}\!\big[\rho_{2:n}\big] &= \sum_j |a_j|
		\operatorname{CanonPur}\!\big[\rho_{2L}^j\big] \otimes \left[ \bigotimes_{i=2}^{n-1} \ket{j}_{(iR)(i^+L)} \bra{j}_{(iR)(i^+L)} \right]
	\\&\mkern100mu \otimes \operatorname{CanonPur}\!\big[\rho_{nR}^j\big] \otimes \left[ \bigotimes_{i=2}^{n-1} \ket{j}_{(\overline{iR})(\overline{i^+L})} \bra{j}_{(\overline{iR})(\overline{i^+L})} \right] .
	\label{eq:SOPS_rho_canonpurif}
\end{aligned} \end{align}
The canonical purifications $\operatorname{CanonPur}\!\big[\rho_{2L}^j\big]$ and $\operatorname{CanonPur}\!\big[\rho_{nR}^j\big]$ live in the Hilbert spaces $\calHconj_{2L}^j \otimes \calH_{2L}^j$ and $\calH_{nR}^j \otimes \calHconj_{nR}^j$ respectively.
After we swap the left/right labels for the splittings (Def.~\ref{def:splitting}) of $\calHconj_i$, Eq.~\eqref{eq:SOPS_rho_canonpurif} takes on the form a SOPS with respect to the decomposition $\big(\calH_2, \calH_3, \dots, \calH_n, \calHconj_n, \calHconj_{n-1}, \dots, \calHconj_2\big)$.
\end{proof}
\end{lemma}

%%%%%%%%%%%%%%%%%%%%
\vspace{3ex}
We now state the main result of the section, the second of the structure theorems.
\begin{theorem}[States with vanishing $h$]\label{thm:zero_h_app}
A pure tripartite quantum state $\ket{\psi}_{ABC}$ is a sum of triangle states (a SOPS with $N=3$) if and only if $h(A:B)=0$ for $\rho_{AB}$.
\end{theorem}
%%%%%%%%%%%%%%%%%%%%

%Next, we show the ``if'' direction of Theorem~\ref{thm:zero_h_app}.  For this, we will need the following lemmas (\ref{lemma:product_state}, \ref{fact:prob_permutation}) and proposition (\ref{prop:4body_cyclicMarkov}).
\noindent To prove the ``only if" direction of Theorem~\ref{thm:zero_h_app}, we use Porism~\ref{lemma:SOPS_cyclicMarkov} and Lemma~\ref{lemma:SOPS_canonpurif}. To prove the ``if" direction, we will additionally need Lemma~\ref{lemma:product_state}, Fact~\ref{fact:prob_permutation}, and Proposition~\ref{prop:4body_cyclicMarkov} below.

\begin{lemma}\label{lemma:product_state}
Given a pure state on four parties $A,B,C,D$, if it is both a product state under the bipartition $(AB,CD)$ and $(AD,BC)$:
\begin{align}
    \ket{ABCD} = \ket{AB}\ket{CD} = \ket{AD}\ket{BC} ,
\end{align}
then it is a product state across all four parties, i.e.,
\begin{align}
    \ket{ABCD} = \ket{A}\ket{B}\ket{C}\ket{D} .
\end{align}
\begin{proof}
Since the state is a product state under the bipartition $(AB,CD)$, then $\rho_{AD} = \rho_A\otimes\rho_D$. The state is also a product state under the bipartition $(AD,BC)$, hence $\rho_{AD}$ is a pure-state density matrix with rank $1$, so $\rho_{A}$ and $\rho_{D}$ must be pure. The same argument applies to $\rho_{B}$ and $\rho_{C}$. Therefore the state is a product state across the four parties.
\end{proof}
\end{lemma}

\begin{fact}\label{fact:prob_permutation}
Let $p_a$ and $q_a$ be two probability distributions on $1 \leq a \leq N$.  Then the sum
\begin{align}
    \sum_{a=1}^N p_a q_a = 1
\end{align}
if and only if the distribtions $p,q$ are identical with a single nonzero entry.
That is, there exist $u$ such that $p_u = q_u = 1$ and $p_a = q_a = 0$ for all $a \neq u$.

\begin{proof}
Because $0 \leq p_a, q_a \leq 1$, the probabilities obey the inequalities
\begin{align}
	\left(\sum_{a=1}^N p_a q_a\right)^2 \leq \left(\sum_{a=1}^N p_a^2\right) \left(\sum_{a=1}^N q_a^2\right)
		\leq \left(\sum_{a=1}^N p_a\right) \left(\sum_{a=1}^N q_a\right) = 1^2 = 1 .
	\label{eq:prob_CauchySchwarz}
\end{align}
The first inequality is the Cauchy-Schwarz inequality, where equality holds iff $p_a \propto q_a$.
The second inequality follows from the fact that $p_a^2 \leq p_a$, where equality holds iff $p_a = 0$ or $1$.

The first term is unity iff both inequalities are satified equality, i.e., all the terms of Eq.~\eqref{eq:prob_CauchySchwarz} are equal.
From the normalization constraint $\sum_a p_a = \sum_a q_a = 1$, this is equivalent to $p_a = q_a$ for all $a$, and that exactly one of $p_a$ is unity, while the remaining $p$'s vanishes.
\end{proof}
\end{fact}

\begin{proposition}\label{prop:4body_cyclicMarkov}
Let $\ket{\psi}_{1234}$ be a pure 4-party state in $\calH_1 \otimes \calH_2 \otimes \calH_3 \otimes \calH_4$.  If $I\big( \calH_{i-1}:\calH_{i+1} \big| \calH_i \big) = 0$ for all $i\in\{1,\dots,4\}$, then $\ket{\psi}_{1234}$ is a SOPS with respect to $(\calH_1,\calH_2,\calH_3,\calH_4)$.

\begin{proof}
We first consider $\calH_1$ and the condition $I(\calH_1:\calH_3|\calH_2)=0$ and use Theorem~\ref{thm:quantum_markov_chain} to observe that there exists a decomposition
\begin{align}
    \calH_2 = \bigoplus_k \calH_{2L}^k \otimes \calH_{2R}^k
    \label{eq:H2_decomposition}
\end{align}
such that the reduced density matrix $\rho_{123}$ on $\calH_1\otimes\calH_2\otimes\calH_3$ may be written as
\begin{align}
    \rho_{123} = \sum_k q^{(2)}_k \, \rho_{1(2L)}^k \otimes \rho_{(2R)3}^k
    \label{eq:rho123}
\end{align}
where $\big\{q^{(2)}_k\big\}$ form a probability distribution.
For convenience, we take $q^{(2)}_k > 0$ for all $k$ (truncating the sum if necessary).
We consider the canonical purification of the reduced density matrix $\rho_{123}$ to $\calHconj_1\otimes\calHconj_2\otimes\calHconj_3$. From Eq.~\eqref{eq:rho123}, the canonical purification is
\begin{align}
    \operatorname{CanonPur}\!\big[\rho_{123}\big]
    = \sum_k \sqrt{q^{(2)}_k} \, \bigket{k}_{1(2L)\overline{1(2L)}} \otimes \bigket{k}_{(2R)3\overline{(2R)3}} \,,
    \label{eq:rho123_purification}
\end{align}
where $\bigket{j}_{1(2L)\overline{1(2L)}} \in \calH_1 \otimes \calH_2 \otimes \calHconj_1 \otimes \calHconj_2$ is the canonical purification of $\rho_{1(2L)}^j$, etc.
Note that the canonical purification may be obtained by isometry acting on $\calH_4$ of the original pure state $\ket{\psi}_{1234}$.
This may alternatively be viewed as identification of a particular basis in $\calH_4$, and this viewpoint is just the difference between active and passive transformation.
As a result, the isometry does not change any entanglement properties among the four parties.
%We may therefore slightly abuse notation and use $\ket{\calH_{1}\calH_{2}\calH_{3}\calH_{4}}$ to denote both the original four-party state and the canonical purification of $\rho_{123}$ to $\calH_4$.
The state $\ket{\psi}_{1234}$ must take the form
\begin{align}
    \ket{\psi}_{1234} = \sum_k \sqrt{q^{(2)}_k} \, \bigket{k}_{(4R)1(2L)} \bigket{k}_{(2R)3(4L)} \,,
    \label{eq:H24_decomp}
\end{align}
%In this purification, just as there is a decomposition for $\calH_1$, we can see that there is also a decomposition of $\calH_3$:
where the decomposition of $\calH_4$ 
\begin{align}
    \calH_4 = \bigoplus_m \calH_{4R}^m \otimes \calH_{4L}^m
    \label{eq:H4_decomposition}
\end{align}
is induced by the reflection of Eq.~\eqref{eq:H2_decomposition} in the canonical purification.
We have interchanged the labeling of $L$ and $R$ for later convenience.

Similarly, from $I(\calH_4:\calH_2|\calH_1)=0$, there exists a splitting on $\calH_1$ and $\calH_3$ such that
\begin{align}
    \ket{\psi}_{1234} = \sum_j \sqrt{q^{(1)}_j} \, \bigket{j}_{(1R)2(3L)} \bigket{j}_{(3R)4(1L)} .
    \label{eq:H13_decomp}
\end{align}

Up to now, we have introduced a splitting on the four Hilbert spaces and the Schmidt decompositions between two different cuts.
Comparing Eqs.~\eqref{eq:H13_decomp} and \eqref{eq:H24_decomp}, it is not clear that there is a term-by-term equivalence.
We now use projectors onto the orthogonal subspaces of each Hilbert space, such as in Eq.~\eqref{eq:H2_decomposition}, to find that there is indeed such an equivalence and that that each term in the sum is precisely a term in Eq.~\eqref{eq:SOPS_definition}.

Let $P^{(i)}_{k}$ be a projector on to the $k$\textsuperscript{th} subspace in the decomposition of $\calH_i$, i.e., $\calH_{iL}^k \otimes \calH_{iR}^k$. As they are projectors onto the Schmidt bases,
$\sum_k P^{(i)}_{k} \ket{\psi}_{1234} = \ket{\psi}_{1234}$.
In addition, from Eqs.~\eqref{eq:H2_decomposition}, \eqref{eq:H24_decomp}, and~\eqref{eq:H4_decomposition}, we have
\begin{subequations} \label{eq:kronecker_1234} \begin{align}
    P^{(2)}_j P^{(4)}_k \ket{\psi}_{1234} = \delta_{jk} P^{(2)}_k \ket{\psi}_{1234} .
    \label{eq:kronecker_24}
\end{align}
%In other words, a projection onto a Schmidt vector may be obtained by projection on either Hilbert space
Similarly, Eq.~\eqref{eq:H13_decomp} implies that
\begin{align}
    P^{(1)}_\ell P^{(3)}_m \ket{\psi}_{1234} = \delta_{\ell m} P^{(1)}_\ell \ket{\psi}_{1234} .
    \label{eq:kronecker_13}
\end{align} \end{subequations}

We iet $\bigket{k,j}_{(4R)1(2L)}$ be the normalized state from the projection $P^{(1)}_j \bigket{k}_{(4R)1(2L)}$;
we rewrite part of Eq.~\eqref{eq:H24_decomp} as
\begin{subequations} \begin{alignat}{2}
    P^{(1)}_j \bigket{k}_{(4R)1(2L)} &= M^{(1)}_{k \to j} \bigket{k,j}_{(4R)1(2L)}
        && \in \calH_{4R}^k \otimes \calH_{1L}^j \otimes \calH_{1R}^j \otimes \calH_{2L}^k \,,
\end{alignat}
where $\sum_j \big|M^{(1)}_{k\to j}\big|^2 = 1$ as the normalization condition.
(If one of the projections $P_j \ket{k}$ vanishes, i.e., $M_{k \to j} = 0$, we may simply assign $\ket{k,j}$ to be an arbitrary normalized state in the respective Hilbert space.)

Likewise, the projection on the Schmidt states in Eqs.~\eqref{eq:H24_decomp} and~\eqref{eq:H13_decomp} are denoted as
\begin{alignat}{2}
    P^{(2)}_k \bigket{\ell}_{(1R)2(3L)} &= M^{(2)}_{\ell \to k} \bigket{\ell,k}_{(1R)2(3L)}
        && \in \calH_{1R}^\ell \otimes \calH_{2L}^k \otimes \calH_{2R}^k \otimes \calH_{3L}^\ell \,,
\\  P^{(3)}_\ell \bigket{m}_{(2R)3(4L)} &= M^{(3)}_{m \to \ell} \bigket{m,\ell}_{(2R)3(4L)}
        && \in \calH_{2R}^m \otimes \calH_{3L}^\ell \otimes \calH_{3R}^\ell \otimes \calH_{4L}^m \,,
\\  P^{(4)}_m \bigket{j}_{(3R)4(1L)} &= M^{(4)}_{j \to m} \bigket{j,m}_{(3R)4(1L)}
        && \in \calH_{3R}^j \otimes \calH_{4L}^m \otimes \calH_{4R}^m \otimes \calH_{1L}^j \,,
\end{alignat} \end{subequations}
with $\sum_k \big|M^{(2)}_{\ell \to k}\big|^2 = \sum_\ell \big|M^{(3)}_{m \to \ell}\big|^2 = \sum_m \big|M^{(4)}_{j \to m}\big|^2 = 1$.

Let $P_{jk\ell m} \equiv P^{(1)}_j P^{(2)}_k P^{(3)}_\ell P^{(4)}_m$ be a product of (commuting) projectors.
We look at its action on $\ket{\psi}_{1234}$ written in two different ways [Eqs.~\eqref{eq:H24_decomp} and~\eqref{eq:H13_decomp}].
By Eqs.~\eqref{eq:kronecker_1234}, any projector for which $j \neq \ell$ or $k \neq m$ will annihilate the state $\ket{\psi}_{1234}$; it therefore suffices to consider only projectors of the form $P_{jkjk}$.
\begin{subequations} \label{eq:Pjkjk_on_psi} \begin{align}
    P_{jkjk} \ket{\psi}_{1234} &= \sqrt{q^{(2)}_k} \, M^{(1)}_{k \to j} M^{(3)}_{k \to j} \, \bigket{k,j}_{(4R)1(2L)} \, \bigket{k,j}_{(2R)3(4L)}
        && \text{[from Eq.~\eqref{eq:H24_decomp}]}, \label{eq:Pjkjk_on_psi13}\\
    P_{jkjk} \ket{\psi}_{1234} &= \sqrt{q^{(1)}_j} \, M^{(2)}_{j \to k} M^{(4)}_{j \to k} \, \bigket{j,k}_{(1R)2(3L)} \, \bigket{j,k}_{(3R)4(1L)}
        && \text{[from Eq.~\eqref{eq:H13_decomp}]}. \label{eq:Pjkjk_on_psi24}
\end{align} \end{subequations}

Consider the application of $P^{(1)}_j P^{(3)}_j$ on $\ket{k}_{(4R)1(2L)} \ket{k}_{(2R)3(4L)}$, individual terms in the RHS of Eq.~\eqref{eq:H24_decomp}.
\begin{align}
	\bigket{k}_{(4R)1(2L)} \bigket{k}_{(2R)3(4L)} =
	\sum_j P^{(1)}_j P^{(3)}_j \bigket{k}_{(4R)1(2L)} \bigket{k}_{(2R)3(4L)} &= \sum_j M^{(1)}_{k \to j} M^{(3)}_{k \to j} \, \bigket{k,j}_{(4R)1(2L)} \, \bigket{k,j}_{(2R)3(4L)} \,.
\end{align}
The normalization condition for the projection is
\begin{align}
	\forall k , \quad \sum_j \big| M^{(1)}_{k \to j} M^{(3)}_{k \to j} \big|^2 = 1
\end{align}
Applying Fact~\ref{fact:prob_permutation} to the probability distributions $\big|M^{(1)}_{k \to j}\big|^2$ and $\big|M^{(3)}_{k \to j}\big|^2$, we conclude that there exist a map $j(k)$, such that $\big|M^{(1)}_{k \to j(k)}\big| = \big|M^{(3)}_{k \to j(k)}\big| = 1$, and that the other terms vanishes.
Likewise, the same argument applied to the individual terms in the RHS of Eq.~\eqref{eq:H13_decomp} implies that there exists a single nonzero entry $M^{(2,4)}_{j\to k(j)}$.
In equation form,
\begin{align}
	\big|M^{(1)}_{k \to j}\big| &= \big|M^{(3)}_{k \to j}\big| = \begin{cases} 1 & j = j(k), \\ 0 & j \neq j(k). \end{cases}
&	\big|M^{(2)}_{j \to k}\big| &= \big|M^{(4)}_{j \to k}\big| = \begin{cases} 1 & k = k(j), \\ 0 & k \neq k(j). \end{cases}
\end{align}
Next, the Eqs.~\eqref{eq:Pjkjk_on_psi13} and~\eqref{eq:Pjkjk_on_psi24} implies that the coefficients have equal absolute values.
This is possible only if $j(k)$ and $k(j)$ are inverse functions: i.e., $j(k(j)) = j$ and $k(j(k)) = k$.  In addition we have $q^{(1)}_j = q^{(2)}_{k(j)}$.

Finally, applying Lemma~\ref{prop:4body_cyclicMarkov} to Eqs.~\eqref{eq:Pjkjk_on_psi}, we can write each $P_{jkjk} \ket{\psi}_{1234}$ as a product of states.
\begin{align}
	\ket{\psi}_{1234}
	&= \sum_{j,k} P_{jkjk} \ket{\psi}_{1234}  \notag
\\	&= \sum_j P_{j,k(j),j,k(j)} \sqrt{q^{(2)}_{k(j)}} \, M^{(1)}_{k(j) \to j} M^{(3)}_{k(j) \to j}
	\, \bigket{k,j}_{(4R)(1L)} \bigket{j,k}_{(1R)(2L)} \bigket{k,j}_{(2R)(3L)} \bigket{j,k}_{(3R)(4L)} \,,
\end{align}
where $\ket{x,y}_{(AR)(BL)} \in \calH_{AR}^x \otimes \calH_{BL}^y$.
This is precisely the form of a SOPS (Def.~\ref{def:SOPS}).
\end{proof}
\end{proposition}

To summarize, we used quantum Markov property [Theorem~\ref{thm:quantum_markov_chain}] to find a decomposition of each local Hilbert space $\calH_i = \bigoplus_\mu \calH_{iL}^\mu \otimes \calH_{iR}^\mu$.
The decomposition gives us orthogonal projectors onto a basis in which the Schmidt vectors along each cut factorized.
Aided by Lemmas~\ref{lemma:product_state} and~\ref{fact:prob_permutation}, we then found that there is a bijection between the Schmidt vectors and values along different cuts.
This enabled us to find orthogonal projectors $P_{jkjk}$ which when applied to $\ket{\psi}_{1234}$, return a polygon state.

\vspace{3ex}
With the bulk of the technical work behind us, we now complete the proof of the structure theorem for $h$.

\begin{proof}[Proof of Theorem~\ref{thm:zero_h_app}]
Let $\ket{\psi}_{ABC}$ be a 3-party state on $A$, $B$, and $C$.
We denote the canonical purification of $\rho_{AB}$ as
\begin{align} \begin{aligned}
	\ket{\Psi}_{AB\bar{B}\bar{A}} &\equiv \operatorname{CanonPur}\!\big[ \rho_{AB} \big] && \in \calH_A \otimes \calH_B \otimes \calHconj_B \otimes \calHconj_A .
\end{aligned} \end{align}
%The canonical purification of $\rho_{AB}$, which we denote by $\ket{\Psi}_{AB\bar{B}\bar{A}}$, lives in the Hilbert space $\calH_A \otimes \calH_B \otimes \calHconj_B \otimes \calHconj_A$.
\noindent We recast $h(A:B)$ as a conditional mutual information:
\begin{align} \begin{aligned}
	h(A:B) &= S_{A\bar{A}} + (S_{AB} - S_A - S_B)
	\\	&= S_{\bar{A}A}+S_{AB}-S_A-S_{\bar{A}AB}
	\\ &= I\big(\bar{A}:B\big|A\big) . \label{eq:h_equals_cmi}
\end{aligned} \end{align}
Because of the symmetry $A \leftrightarrow B$ and $A,B \leftrightarrow \bar{A},\bar{B}$, the following quantities are equal.
\begin{align}
	h(A:B) = I\big(\bar{A}:B\big|A\big) = I\big(A:\bar{B}\big|B\big) = I\big(B:\bar{A}\big|A\big) = I\big(\bar{B}:A\big|\bar{A}\big) .
	\label{eq:h_equals_cmi4}
\end{align}

If $\ket{\psi}_{ABC}$ is a SOTS (SOPS with $N=3$), then by Lemma~\ref{lemma:SOPS_canonpurif} the canonical purification $\ket{\Psi}_{AB\bar{B}\bar{A}}$ is a SOPS with respect to $(\calH_A,\calH_B,\calHconj_B,\calHconj_A)$.
Via Porism~\ref{lemma:SOPS_cyclicMarkov}, Eq.~\eqref{eq:h_equals_cmi4} is identically zero. Thus, if $\ket{\psi}_{ABC}$ is a SOTS, $h(A:B)=0$.

The proof for the converse statement is as follows.
If $h(A:B) = 0$ for $\ket{\psi}_{ABC}$, then Eq.~\eqref{eq:h_equals_cmi4} vanishes for $\ket{\Psi}_{AB\bar{B}\bar{A}}$.
It follows from Proposition~\ref{prop:4body_cyclicMarkov} that $\ket{\Psi}_{AB\bar{B}\bar{A}}$ is a 4-party SOPS.
Finally observe that $\ket{\psi}_{ABC}$ is isometric to $\ket{\Psi}_{AB\bar{B}\bar{A}}$ after coarse-graining $\bar{B}\bar{A} \to C$, which by Lemma~\ref{lemma:SOP_CG} makes $\ket{\psi}_{ABC}$ a 3-party SOPS.

This complete the proof that $h(A:B) = 0$ if and only if $\ket{\psi}_{ABC}$ is a SOTS.
\end{proof}

\noindent \textbf{Remark.}  The converse of Porism~\ref{lemma:SOPS_cyclicMarkov} is the statement
\begin{align}
	\forall i\quad I(\calH_{i-1}:\calH_{i+1}|\calH_i)=0 \quad\overset{?}{\Longrightarrow}\quad \ket{\psi}_{1\dots N} \in \mathrm{SOPS}(\calH_1,\dots,\calH_N) .
	\label{eq:Lemma_SOPS_cyclicMarkov_converse}
\end{align}
Proposition~\ref{prop:4body_cyclicMarkov} proves that this is indeed true for $N=4$.
The statement is trivially true for $N=3$, since for three parties the left-hand-side implies that $\ket{\psi}_{123}$ is a product state.
Evidently this statement is false for $N\geq6$, because of the existence of a 6-party perfect tensor. The case for $N=5$ remains an open problem.
%\RM{\bf Karthik will tell us if Eq.~\eqref{eq:Lemma_SOPS_cyclicMarkov_converse} is true for $N=5$. (Or not)}
%(At the end of App.~\ref{app:gapped_systems}, we provide a more stringent criteria which should be sufficient to show that a state is an $N$-party SOPS.)

\section{Coarse-graining of matrix product states}\label{app:coarse_graining}
Here we present the details of the matrix product state techniques that are used to compute $E_P(A:B)$ and $S_R(A:B)$ for a critical quantum spin chain. As mentioned in the main text, we coarse-grain spins in regions $A,B,C$ and truncate the Hilbert spaces $\calH_A,\calH_B,\calH_C$ into $\calH_{\tilde{A}},\calH_{\tilde{B}},\calH_{\tilde{C}}$. The starting point is a periodic uniform matrix product state (puMPS) that represents the ground state. A puMPS is composed of $N$ identical rank-3 tensors $M$,
\begin{equation}
\label{eq:original_MPS}
    \ket{\psi(M)} =\sum_{s_1 s_2\cdots s_n}\Tr(M_{s_1}M_{s_2}\cdots M_{s_N})\ket{s_1s_2\cdots s_N},
\end{equation}
where $M_{s_i}$ is a $ D\times D$ matrix, $s_i=1,2,\dots,d$ is the index for the Hilbert space on one site, and $D$ is the bond dimension. The bond dimension $D$ restricts the amount of entanglement in the ansatz. Specifically, the reduced density matrix of a puMPS on any contiguous region has rank at most $D^2$. In order to represent the ground state faithfully, $D$ grows polynomially with $N$ for a critical spin chain, but stays constant for a gapped spin chain \cite{Verstraete2006}. We employ methods in Ref.~\onlinecite{Zou2018} to minimize the energy with respect to the Hamiltonian $H$ and obtain the optimized puMPS $\ket{\psi(M)}$ as an approximation to the ground state.

\begin{figure}[h]
    \centering
    \includegraphics[width=0.5\linewidth]{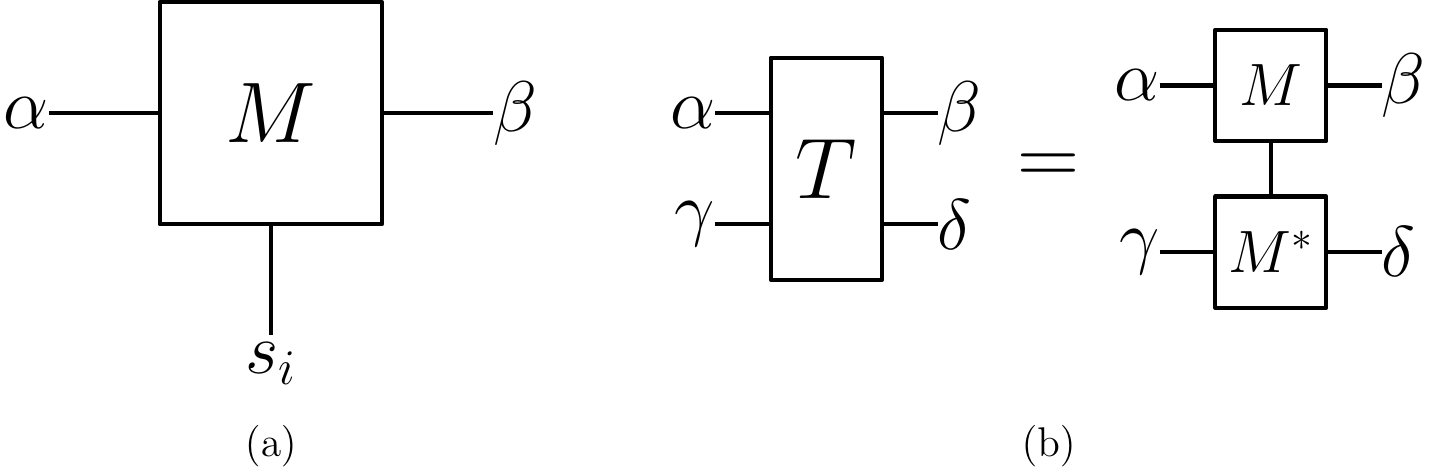}
    \caption{(a) Unit cell tensor $M$. Virtual indices are labeled $\alpha,\beta$ and physical index is labeled $s_i$. (b) $T_{\alpha\gamma, \beta \delta}$ is formed by the contraction of the unit cell tensors along the physical index.}
    \label{fig:tensor_definitions}
\end{figure}

Next, we block spins to obtain a series of tensors $M^{(n)}$ that represents the coarse-grained tensor for $n$ contiguous sites. Consider the transfer matrix
\begin{equation}
\label{eq:TM_MPS}
    T_{\alpha\beta\gamma\delta}=\sum_{s=1}^d M_{s\alpha\beta}M^{*}_{s\gamma\delta}.
\end{equation}
For clarity, the indices and multiplication are shown using graphical tensor notation as well in Fig.~\ref{fig:tensor_definitions}. Grouping the indices $(\alpha,\gamma)$ and $(\beta,\delta)$, this tensor can be viewed as a $D^2\times D^2$ matrix. We can take its $n$-th power $T^n$, e.g.
\begin{equation}
    (T^2)_{\alpha\gamma,\beta\delta}=\sum_{\beta_1,\delta_1} T_{\alpha \beta_1 \gamma \delta_1} T_{\beta_1 \beta \delta_1\delta}. 
\end{equation}
Now we regroup the indices $(\alpha,\beta)$ and $(\gamma,\delta)$ for each tensor $(T^n)_{\alpha\beta\gamma\delta}$ to obtain a $D^2\times D^2$ matrix (different from above) that is Hermitian and positive. These matrices admit an eigenvalue decomposition,
\begin{equation}
\label{eq:TMn_MPS}
    (T^n)_{\alpha\beta,\gamma\delta}=\sum_{S=1}^{d^{(n)}} \lambda^{(n)}_{S} U^{(n)}_{S\alpha\beta}U^{(n)*}_{S\gamma\delta},
\end{equation}
where $d^{(n)}$ is the number of positive eigenvalues. It can be easily shown that $d^{(n)}\leq \min{(D^2,d^n)}$. Let
\begin{equation}
\label{eq:coarse_grain_eigv}
    M^{(n)}_{S\alpha\beta}=\sum_{S=1}^{d^{(n)}}\sqrt{\lambda^{(n)}_{S}} U^{(n)}_{S\alpha\beta},
\end{equation}
then the $n$th power of the transfer matrix can be recovered by $M^{(n)}_{S\alpha\beta}$,
\begin{equation}
   (T^n)_{\alpha\beta\gamma\delta}=\sum_{S=1}^{d^{(n)}} M^{(n)}_{S\alpha\beta}M^{(n)*}_{S\gamma\delta}.
   \label{eq:coarse_M}
\end{equation}
So far we have obtained a series of tensors $M^{(n)}$ as a result of coarse-graining $n$ spins.

Now we can compute the mutual information $I(A:B)=S_A+S_B-S_{AB}$. First, we would like to compute $S_A$, where $A$ contains $N_A$ contiguous sites, and the complement $\bar{A}$ contains $N-N_A$ sites. We can use the tensors $M^{(N_A)}$ and $M^{(N_{\bar{A}})}$ to construct a coarse-grained state
\begin{equation}
\label{eq:bipartiteMPS}
    \ket{\tilde{\psi}}_{A\bar{A}}= \sum_{S_{A},S_{\bar{A}}} \sum_{\alpha,\beta=1}^{D} M^{(N_A)}_{S_A \alpha\beta} M^{(N_{\bar{A}})}_{S_{\bar{A}}\beta\alpha}\ket{S_A S_{\bar{A}}}
\end{equation}
A Schmidt decomposition on $\ket{\psi}_{A\bar{A}}$ gives the entanglement spectrum $\{\lambda^2_j\}$ between $A$ and $\bar{A}$, which amounts to a singular value decompostion,
\begin{equation}
\label{eq:coarse_grained_MPS_SVD}
    M^{(N_A)}_{S_A \alpha\beta} M^{(N_{\bar{A}})}_{S_{\bar{A}}\beta\alpha}=\sum_{j} U_{S_A j} \lambda_j V^{*}_{S_{\bar{A}}j},
\end{equation}
where
\begin{equation}
    \sum_{S_A} U_{S_A j} U^{*}_{S_A k}=\delta_{jk},~~ \sum_{S_{\bar{A}}}V^{*}_{S_{\bar{A}}j} V_{S_{\bar{A}}k}=\delta_{jk},
\end{equation}
and $\lambda_j\geq 0$. Note that $U$ depends on $N_A$, and we do not explicitly show the dependence in the notation. Also note that the $U$ here is not to be confused with the $U^{(n)}$ in Eq.~\eqref{eq:coarse_grain_eigv}. The entanglement entropy between $A$ and $\bar{A}$ is
\begin{equation}
    S_{A}=-\sum_{j} \lambda^{2}_j\log(\lambda^2_j).
\end{equation}
Repeating the same procedure with $N_A$ substituted by $N_B$ or $N_{A}+N_{B}$ gives the entanglement entropy $S_B$ or $S_{AB}$. We thus obtain the mutual information $I(A:B)=S_A+S_B-S_{AB}$. For later convenience, we further truncate the physical dimension of $M^{(N_A)}$ using the Schmidt vectors $U_{S_A j}$,
\begin{equation}
\label{eq:MPS_truncation}
    \tilde{M}^{(N_A)}_{\tilde{S}_A\alpha\beta}=\sum_{S_A} U^{*}_{S_A \tilde{S}_A} M^{(N_A)}_{S_A \alpha\beta},
\end{equation}
where we have restricted the index $\tilde{S}_A$ such that
\begin{equation}
    \lambda_{\tilde{S}_A}>\epsilon,~~\forall \tilde{S}_A,
\end{equation}
given an error threshold $\epsilon>0$. The physical dimension is truncated to $\tilde{d}^{(N_A)}$, the number of Schmidt coefficients $\{\lambda_j\}$ larger than a threshold $\epsilon$. By virtue of the singular value decomposition, the truncation only keeps the Schmidt vectors with Schmidt values larger than $\epsilon$. Again, the truncation can be done for any $M^{(n)}$ with $1\leq n \leq N$, by substituting the $N_A$ above with $n$.

Given the coarse-grained tensor $\tilde{M}^{(n)}$ and $N_A,N_B,N_C$, we are ready to construct the coarse-grained tripartite state as a MPS with three sites,
\begin{equation}
\label{eq:truncated_MPS}
\ket{\tilde{\psi}}_{\tilde{A}\tilde{B}\tilde{C}}=\sum_{\tilde{S}_A,\tilde{S}_B,\tilde{S}_C}\sum_{\alpha\beta\gamma} \tilde{M}^{(N_A)}_{\tilde{S}_A\alpha\beta}\tilde{M}^{(N_B)}_{\tilde{S}_B\beta\gamma}\tilde{M}^{(N_C)}_{\tilde{S}_C\gamma\alpha}\ket{\tilde{S}_A \tilde{S}_B \tilde{S}_C}.
\end{equation}
This is the state that is used to compute $E_P(A:B)$ and $S_R(A:B)$ in the main text. The whole renormalization procedure is summarized in Fig.~\ref{fig:MPS_RG}. 

\begin{figure}
    \centering
    \begin{minipage}{0.98\linewidth}
        \begin{minipage}[t]{15mm}\raisebox{22mm}{\subfigure[]{\label{fig:pumps_a}}}\end{minipage}
        \includegraphics[width=0.5\linewidth]{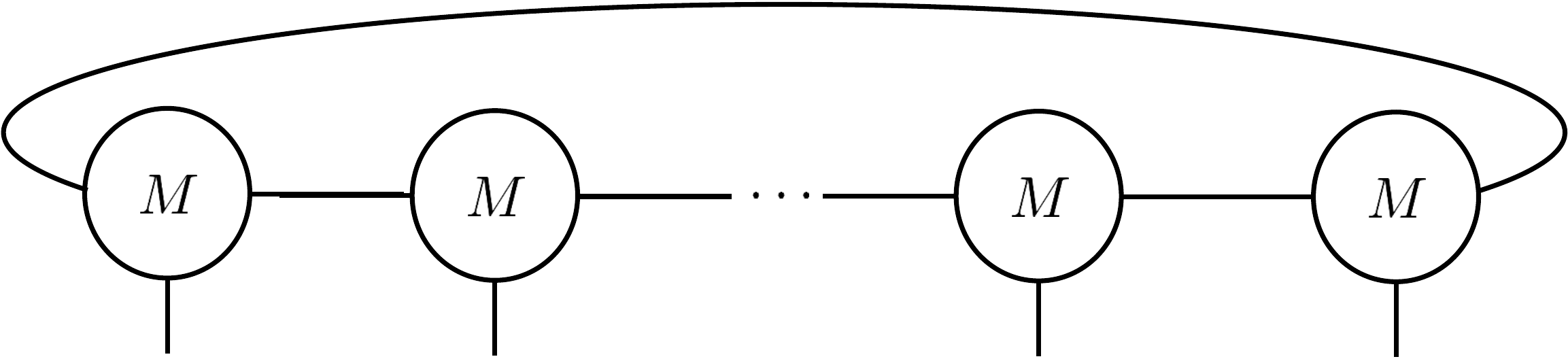}
    \end{minipage}
\\[1ex]
    \begin{minipage}{0.98\linewidth}
        \begin{minipage}[t]{0mm}\hspace{-15mm}\raisebox{36mm}{\subfigure[]{\label{fig:pumps_b}}}\end{minipage}
        \includegraphics[width=0.45\linewidth]{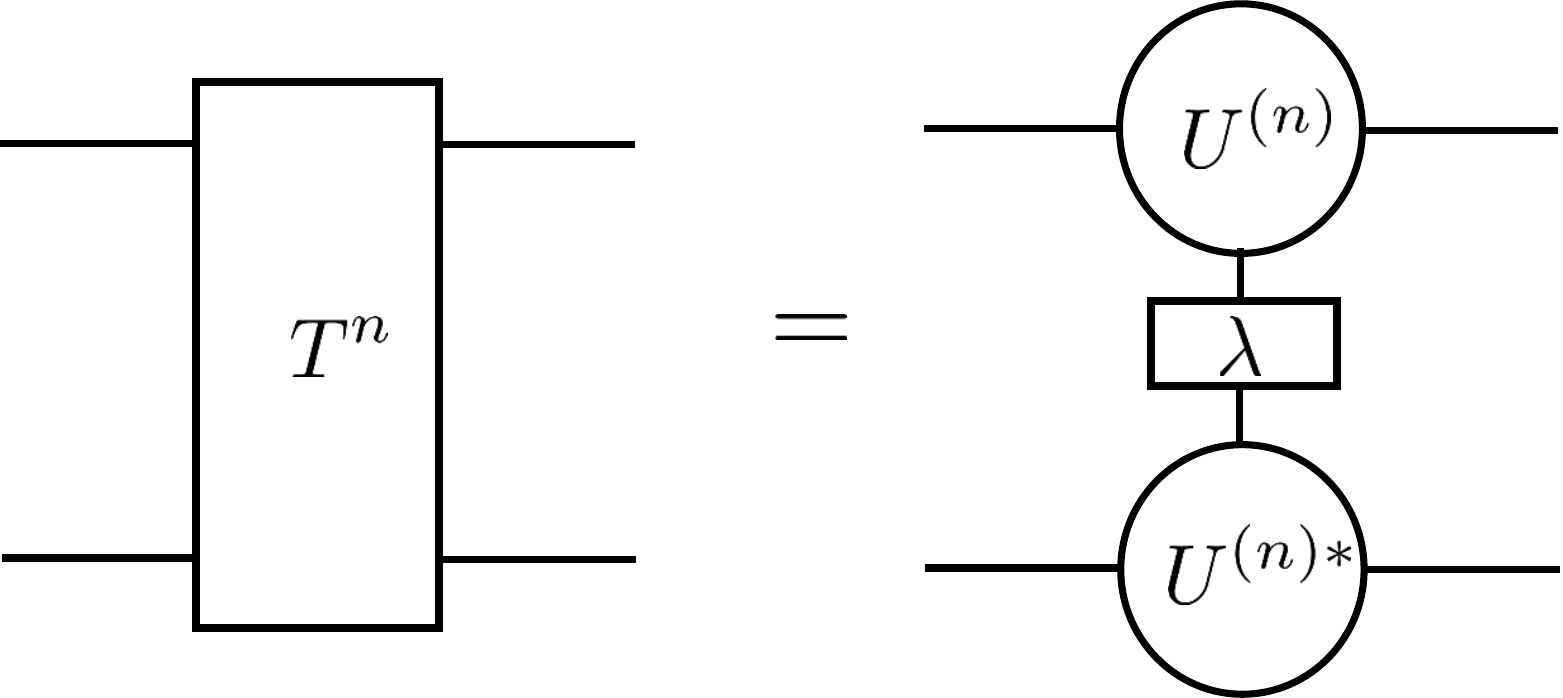}
    \end{minipage}
\\[1ex]
    \begin{minipage}{0.98\linewidth}
        \begin{minipage}[t]{0mm}\hspace{-15mm}\raisebox{25mm}{\subfigure[]{\label{fig:pumps_c}}}\end{minipage}
        \includegraphics[width=0.45\linewidth]{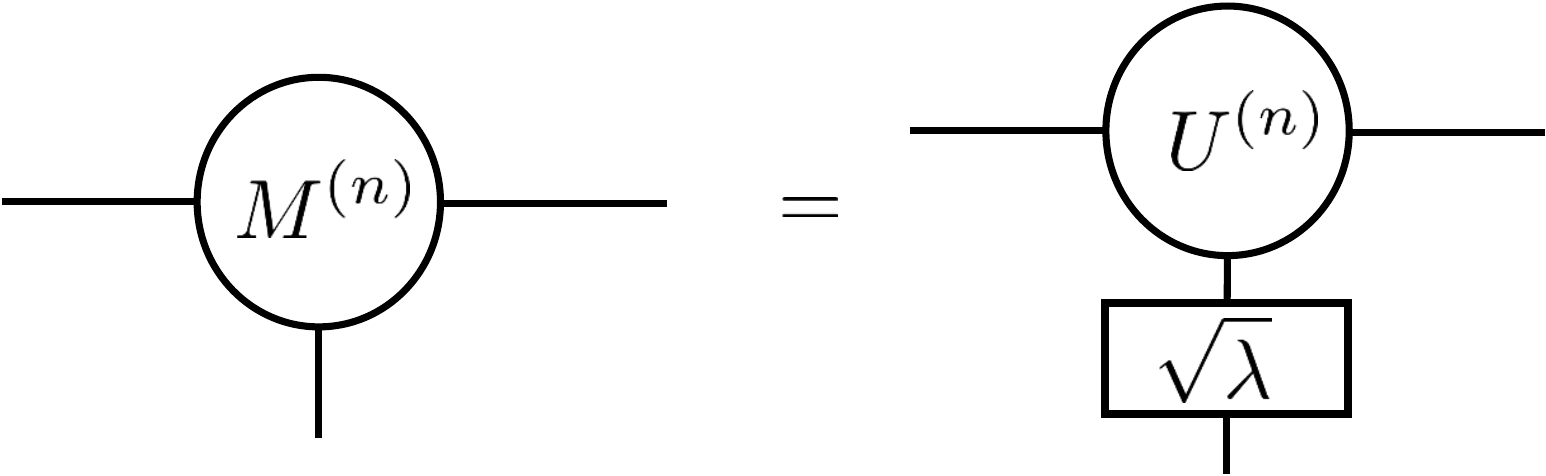}
    \end{minipage}    
\\[1ex]
    \begin{minipage}{0.98\linewidth}
        \begin{minipage}[t]{0mm}\hspace{-15mm}\raisebox{27mm}{\subfigure[]{\label{fig:pumps_d}}}\end{minipage}
        \includegraphics[width=0.6\linewidth]{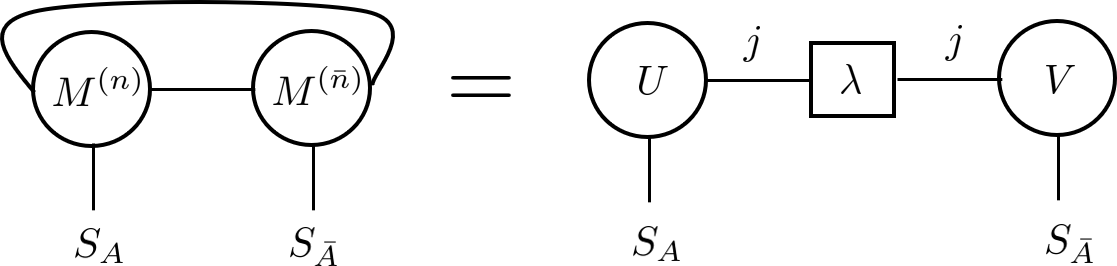}
    \end{minipage}
\\[1ex]
    \begin{minipage}{0.98\linewidth}
        \begin{minipage}[t]{0mm}\hspace{-15mm}\raisebox{30mm}{\subfigure[]{\label{fig:pumps_e}}}\end{minipage}
        \includegraphics[width=0.5\linewidth]{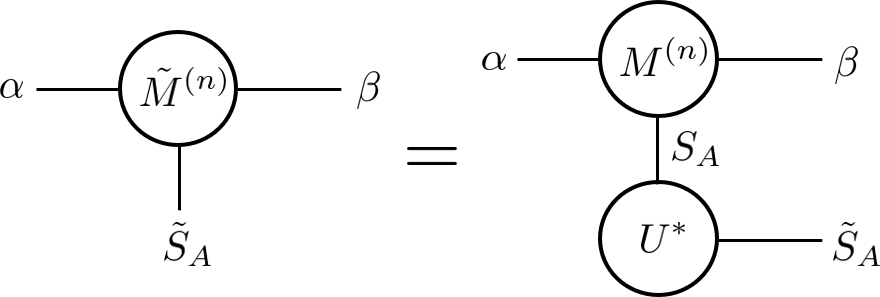}
    \end{minipage}
\\[1ex]
    \begin{minipage}{0.98\linewidth}
        \begin{minipage}[t]{0mm}\hspace{-15mm}\raisebox{30mm}{\subfigure[]{\label{fig:truncated_MPS}}}\end{minipage}
        \includegraphics[width=0.4\linewidth]{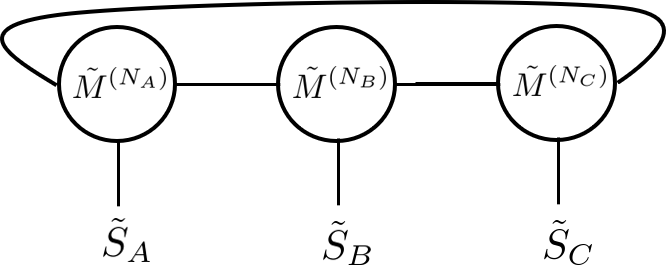}
    \end{minipage}
%    \begin{subfigure}{0.98\linewidth}
%        \includegraphics[width=0.5\linewidth]{RG0crop.pdf}
%        \subcaption{}
%    \end{subfigure}
    %\begin{subfigure}{0.45\linewidth}
    %    \includegraphics[width=0.98\linewidth]{RG1pcrop.pdf}
    %   \subcaption{}
    %\end{subfigure}
%    \begin{subfigure}{0.98\linewidth}
%        \includegraphics[width=0.45\linewidth]{RG2pcrop.pdf}
%        \subcaption{}
%    \end{subfigure}
%    \begin{subfigure}{0.98\linewidth}
%        \includegraphics[width=0.45\linewidth]{RG3pcrop.pdf}
%        \subcaption{}
%    \end{subfigure}
%    \begin{subfigure}{0.98\linewidth}
%        \includegraphics[width=0.6\linewidth]{Mn_truncate.png}
%        \subcaption{}
%    \end{subfigure}
%    \begin{subfigure}{0.98\linewidth}
%        \includegraphics[width=0.5\linewidth]{Mn_truncate2.png}
%        \subcaption{}
%    \end{subfigure}
%    \begin{subfigure}{0.98\linewidth}
%        \includegraphics[width=0.4\linewidth]{psi_tilde.png}
%        \subcaption{}. 
%        \label{fig:truncated_MPS}
%    \end{subfigure}
    \caption{Coarse-graining of the periodic uniform matrix product state. (a) The original MPS as the ground state of the spin chain Hamiltonian, Eq.~\eqref{eq:original_MPS}. %(b) The transfer matrix $T$ of the MPS, Eq.~\eqref{eq:TM_MPS}. 
    (b) The eigenvalue decomposition of the $n$-th power of the transfer matrix $T^n$, Eq.~\eqref{eq:TMn_MPS}. (c) The coarse-grained tensor $M^{(n)}$ as a tensor for $n$ sites, Eq.~\eqref{eq:coarse_grain_eigv}. (d) The Schmidt decomposition of the coarse-grained MPS with respect to $A$ and its complement $\bar{A}$ with $n$ and $\bar{n}=N-n$ sites, respectively, Eq.~\eqref{eq:coarse_grained_MPS_SVD}. (e) Truncation of physical dimensions of the coarse-grained MPS tensors, Eq.~\eqref{eq:MPS_truncation}. (f) The final state $\ket{\tilde{\psi}}$ in Eq.~\eqref{eq:truncated_MPS}.}
    \label{fig:MPS_RG}
\end{figure}

\section{Gradient optimization for \texorpdfstring{$E_P$}{Ep}}\label{app:gradient_descent_ep}
We consider a tripartite state $\ket{\psi}_{ABC}$ where $C$ is further split into $C_L$ and $C_R$, $\calH_C=\calH_{C_L}\otimes \calH_{C_R}$.
Since $\ket{\psi}_{ABC_LC_R}$ is a purification of the reduced density matrix $\rho_{AB}$, all states of the form
\begin{equation}
    \ket{\psi(U_{C_LC_R})}=U_{C_LC_R}\ket{\psi}_{ABC_LC_R}
\end{equation}
gives purifications of $\rho_{AB}$, where $U_{C_LC_R}$ is a unitary operator on $C_LC_R$. Assume that the optimal purification can be achieved with the prescribed Hilbert space $H_{C_L}$ and $H_{C_R}$, then 
\begin{equation}
    E_P(A:B)=\min_{U_{C_LC_R}} S_{AC_L:BC_R}(\ket{\psi(U_{C_LC_R})}).
\end{equation}
We will find the minimum by a gradient optimization. The gradient optimization requires the gradient of the objective function over the argument. To compute the gradient, we first express the reduced density matrix $\rho_{AC_L}$ as
\begin{equation}
\label{eq:rho_ACL}
\rho_{AC_L}=\Tr_{BC_R} \rho,
\end{equation}
where
\begin{equation}
\label{eq:rho}
\rho=U_{C_LC_R}\ket{\psi}\bra{\psi}U^{\dagger}_{C_LC_R}.
\end{equation}
The entanglement entropy is
\begin{equation}
\label{eq:SACL}
S_{AC_L:BC_R}=-\Tr_{AC_L} (\rho_{AC_L} \log \rho_{AC_L}),
\end{equation}
where $\rho_{AC_L}$ depends on $U_{C_LC_R}$ by Eqs.~\eqref{eq:rho_ACL},~\eqref{eq:rho}.

Let $s$ be the label for a step of the gradient optimization. Initially at $s=0$ we have
\begin{equation}
    U_{C_LC_R}(s=0)=\mathbf{1}_{C_LC_R},
\end{equation}
which amounts to picking the original state $\ket{\psi}_{ABC_LC_R}$ as the purification. At each step of gradient optimization, We perform an update of $U_{C_LC_R}$ of the form
\begin{equation}
\label{eq:UpdateU}
    U_{C_LC_R}(s+1)= e^{i\Theta_{C_LC_R}\delta t} U_{C_LC_R}(s),
\end{equation}
where
$\Theta_{C_LC_R}$ is an Hermitian operator on $\calH_{C_L}\otimes \calH_{C_R}$ and $\delta t||\Theta_{C_LC_R}||\ll 1$. Up to higher order terms in $\delta t$, the change in the entanglement entropy $S_{AC_L:BC_R}$ is
\begin{align}
\delta S_{AC_L:BC_R} &= -\Tr_{AC_L} (\delta \rho_{AC_L} \log \rho_{AC_L}) \nonumber \\
&= -i\delta t \Tr_{AC_L} (\Tr_{BC_R}([\Theta_{C_LC_R}, \rho]) \log \rho_{AC_L}) \nonumber \\
&= -i\delta t\Tr([\Theta_{C_LC_R},\rho] \log \rho_{AC_L}\otimes \mathbf{1}_{BC_R}) \nonumber \\
&= -i\delta t \Tr(\Theta_{C_LC_R} [\rho, \log \rho_{AC_L}\otimes \mathbf{1}_{BC_R}]) \nonumber \\
&= \delta t \Tr_{C_LC_R}(\Theta_{C_LC_R}E_{C_LC_R}),
\end{align}
where 
\begin{equation}
\label{eq:envHC}
    E_{C_LC_R}=-i\Tr_{AB}([\rho, \log \rho_{AC_L}\otimes \mathbf{1}_{BC_R}]).
\end{equation}
In the first line we have differentiated Eq.~\eqref{eq:SACL} and used $\Tr_{AC_L}(\delta \rho_{AC_L})=0$ since $\Tr_{AC_L}\rho_{AC_L}\equiv 1$, in the second line we have used the Heisenberg evolution of density matrix $\rho$ and traced out $BC_R$, in the third line we have rearranged the tracings into an overall tracing on the full Hilbert space, in the fourth line we have used the cyclic property of trace, and in the last line we have rearranged the tracing.

If we use the gradient descent algorithm, we choose $\Theta_{C_LC_R}$ to be
\begin{equation}
\label{eq:g_grad}
\Theta_{C_LC_R}= -E_{C_LC_R}.
\end{equation}
In order to determine $\delta t$, we perform a line search to find the $\delta t$ that minimizes $S_{AC_L:BC_R}$, given the update rule Eq.~\eqref{eq:UpdateU} and the gradient direction Eq.~\eqref{eq:g_grad}. We then obtain $U_{C_LC_R}(s+1)$ which can be substituted into Eq.~\eqref{eq:rho} and Eq.~\eqref{eq:envHC} to compute the gradient direction $E_{C_LC_R}$ for the next step of the update. The gradient optimization goes so on and so forth, until the norm of gradient $||E_{C_LC_R}||$ is smaller than some tolerance $\eta$. In a typical gradient descent optimization, the error is quadratic in the norm of gradient. In this work we choose $\eta=10^{-4}$ such that the error in $E_P(A:B)$ is small compared to the finite-size corrections. In practice, we use the nonlinear conjugate gradient (NLCG) method instead of the simple gradient descent. The search direction $\Theta_{C_LC_R}$ in NLCG is a suitable linear combination of the gradient and the search direction in the previous step of iteration.

The computation of Eq.~\eqref{eq:envHC} is the most expensive step in the optimization. Given the tripartite state $\ket{\psi}_{\tilde{A}\tilde{B}\tilde{C}}$ in Fig.~\ref{fig:MPS_RG}, we follow the steps below to compute Eq.~\eqref{eq:envHC}. First, we contract the tensor network in Fig.~\ref{fig:truncated_MPS}, resulting in a three-leg tensor in Fig.~\ref{fig:reshape_psi}, where we have omitted the tilde to simplify the notation. Then we split the leg $C$ into $C_L$ and $C_R$ as prescribed by the decomposition of the Hilbert space. In order to find $\log \rho_{AC_L}$, we first do the Schmidt decomposition with respect to $AC_L$ and $BC_R$, as shown in Fig.~\ref{fig:SVD_ACL}. Then $\log \rho_{AC_L}$ can be represented by Fig.~\ref{fig:logrhoACL}. The density matrix $\rho$ is shown in Fig.~\ref{fig:rho}. We can then compute $\Tr_{AB}(\rho (\log \rho_{AC_L}\otimes \mathbf{1}_{BC_R}))$ by contracting the tensor network in fig,~\ref{fig:gradient_Ep_MPS}. Finally, Eq.~\eqref{eq:envHC} can be computed by
\begin{equation}
\label{eq:envHC_expand}
    E_{C_LC_R}=-i[\Tr_{AB}(\rho (\log \rho_{AC_L}\otimes \mathbf{1}_{BC_R}))-h.c.],
\end{equation}
where $h.c.$ denotes the Hermitian conjugate of $\Tr_{AB}(\rho (\log \rho_{AC_L}\otimes \mathbf{1}_{BC_R}))$.
\begin{figure}
    \centering
    \begin{minipage}{0.98\linewidth}
        \begin{minipage}[t]{0mm}\hspace{-20mm}\raisebox{24mm}{\subfigure[]{\label{fig:reshape_psi}}}\end{minipage}
        \includegraphics[width=0.4\linewidth]{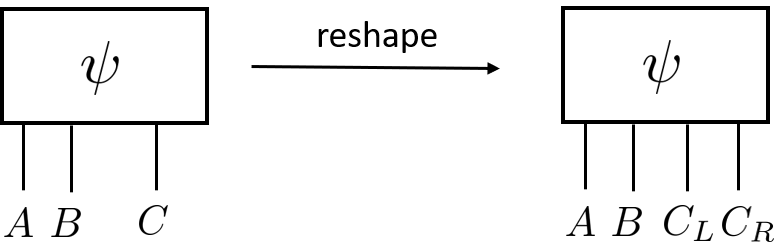}
    \end{minipage}
\\[1ex]
    \begin{minipage}{0.98\linewidth}
        \begin{minipage}[t]{0mm}\hspace{-20mm}\raisebox{20mm}{\subfigure[]{\label{fig:SVD_ACL}}}\end{minipage}
        \includegraphics[width=0.5\linewidth]{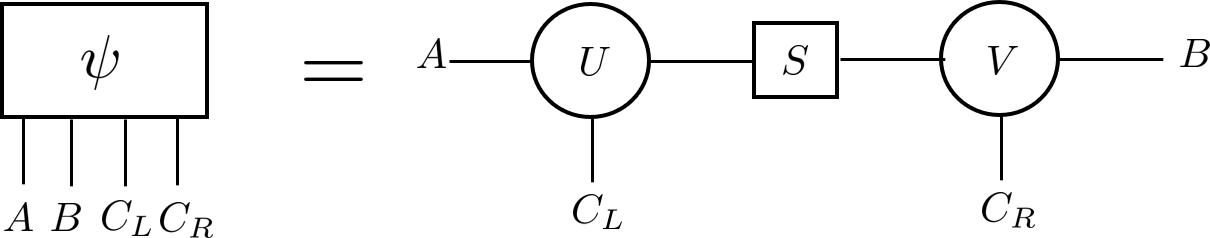}
    \end{minipage}
\\[1ex]
    \begin{minipage}{0.98\linewidth}
        \begin{minipage}[t]{0mm}\hspace{-20mm}\raisebox{34mm}{\subfigure[]{\label{fig:rho}}}\end{minipage}
        \includegraphics[width=0.5\linewidth]{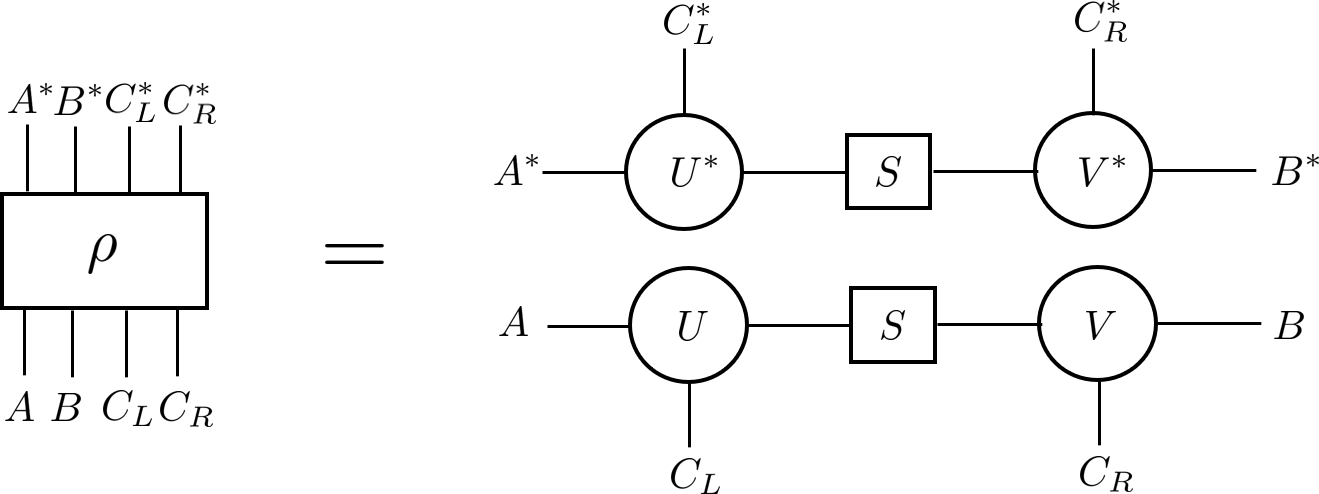}
    \end{minipage}
\\[1ex]
    \begin{minipage}{0.98\linewidth}
        \begin{minipage}[t]{0mm}\hspace{-20mm}\raisebox{37mm}{\subfigure[]{\label{fig:logrhoACL}}}\end{minipage}
        \includegraphics[width=0.6\linewidth]{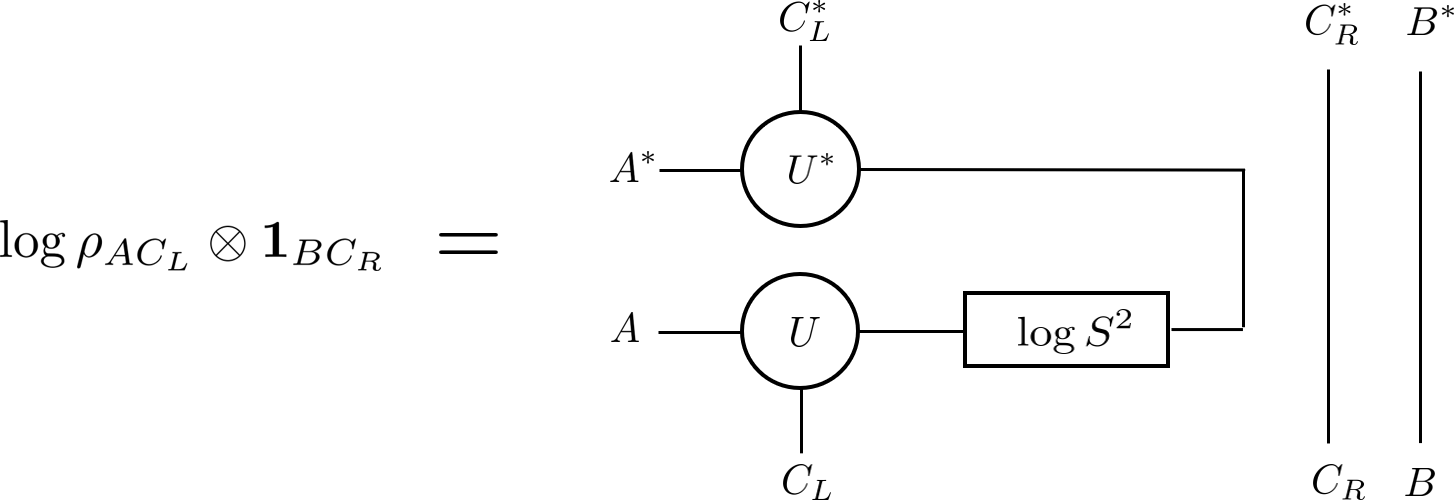}
    \end{minipage}
\\[1ex]
    \begin{minipage}{0.98\linewidth}
        \begin{minipage}[t]{0mm}\hspace{-20mm}\raisebox{33mm}{\subfigure[]{\label{fig:gradient_Ep_MPS}}}\end{minipage}
        \includegraphics[width=0.6\linewidth]{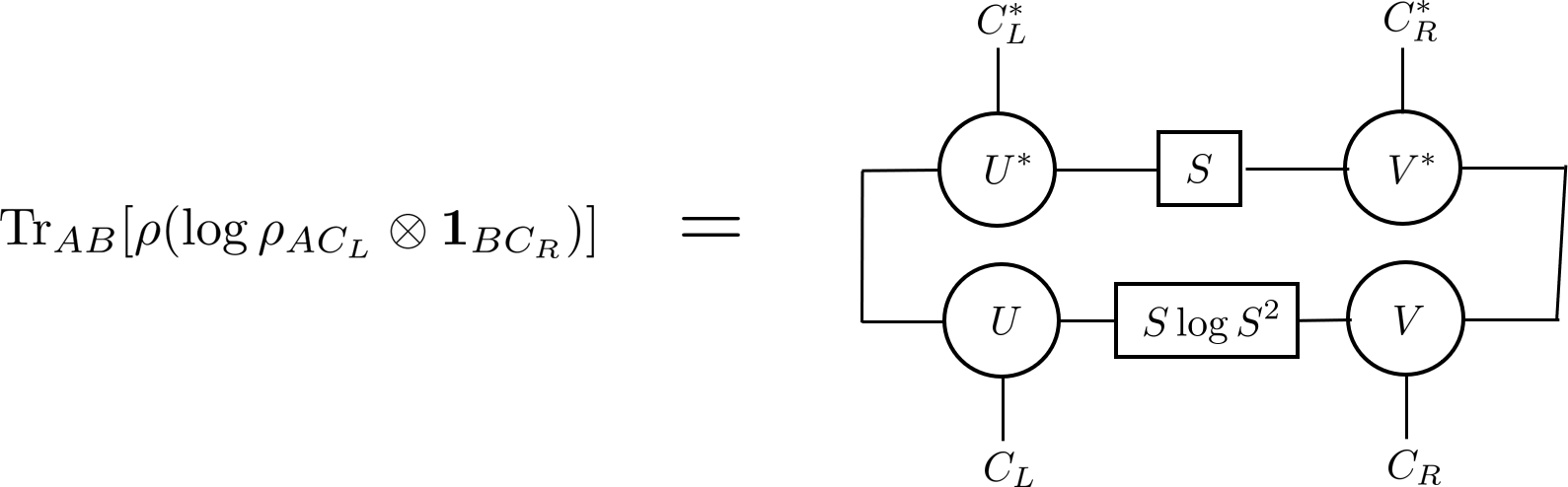}
    \end{minipage}
    \caption{Computation of the gradient Eq.~\eqref{eq:envHC}. (a) Given a bipartition $\calH_{C}=\calH_{C_L}\otimes \calH_{C_R}$, reshape the state $\ket{\psi}$ into a rank 4 tensor. (b) The Schmidt decomposition of the state $\ket{\psi}$ bipartite into $AC_L$ and $BC_R$. (c) The density matrix $\rho$ of the state $\ket{\psi}$.(d) The logarithm of the reduced density matrix $\rho_{AC_L}$. (e) The first term in the square bracket of Eq.~\eqref{eq:envHC_expand}. }
\end{figure}

\section{Different subregion sizes}
First, we argue that $g(A:B)$ and $h(A:B)$ will be independent of the sizes of $A$ and $B$ in the thermodynamic limit. In the thermodynamic limit, the quantum spin chain is described by a 1+1D conformal field theory (CFT). As shown in Refs.~\cite{Nguyen2018,Dutta2019}, the UV divergences in $2E_P(A:B)$, $S_R(A:B)$, and $I(A:B)$ are of the same form -- they scale with the UV cutoff $\Lambda$ as
\begin{equation}
    2E_P(A:B), S_R(A:B), I(A:B) \sim \frac{c^{\CFT}}{3}\log \Lambda.
\end{equation}
The UV divergences in the quantities $g(A:B)\equiv 2E_P(A:B)-I(A:B)$ and $h(A:B)\equiv S_R(A:B)-I(A:B)$ should therefore cancel, making them scale-invariant. In a conformal field theory a scale-invariant quantity is also conformally invariant. In 1+1D, a change in the length of the regions can be implemented by conformal transformations, which includes rescaling the space with arbitrary local weights. Therefore, in a CFT, $g(A:B)$ and $h(A:B)$ do not depend on the sizes of $A$ and $B$. We then expect that on the lattice, $g(A:B)$ and $h(A:B)$ also do not depend on the sizes of $A$ and $B$, once the thermodynamic limit is taken.

We study how $g(A:B)$ and $h(A:B)$ depend on subregion sizes using the O'Brien-Fendley model at $\lambda=0.3$. It is in the Ising universality class but has larger finite-size effect than the Ising model, and we can see the finite-size corrections in $g(A:B)$ and $h(A:B)$ more easily.  We fix ratios $(r_A,r_B,r_C)=(N_A/N, N_B/N, N_C/N)$ that determine the relative sizes and then take the thermodynamic limit $N\rightarrow\infty$. Results with different ratios are shown in Fig.~\ref{fig:OBF}.
\begin{figure}[t!]
    \begin{minipage}{0.49\linewidth}
        \subfigure[]{\label{fig:OBFgs}}\\[-1ex]
        \includegraphics[width=\linewidth]{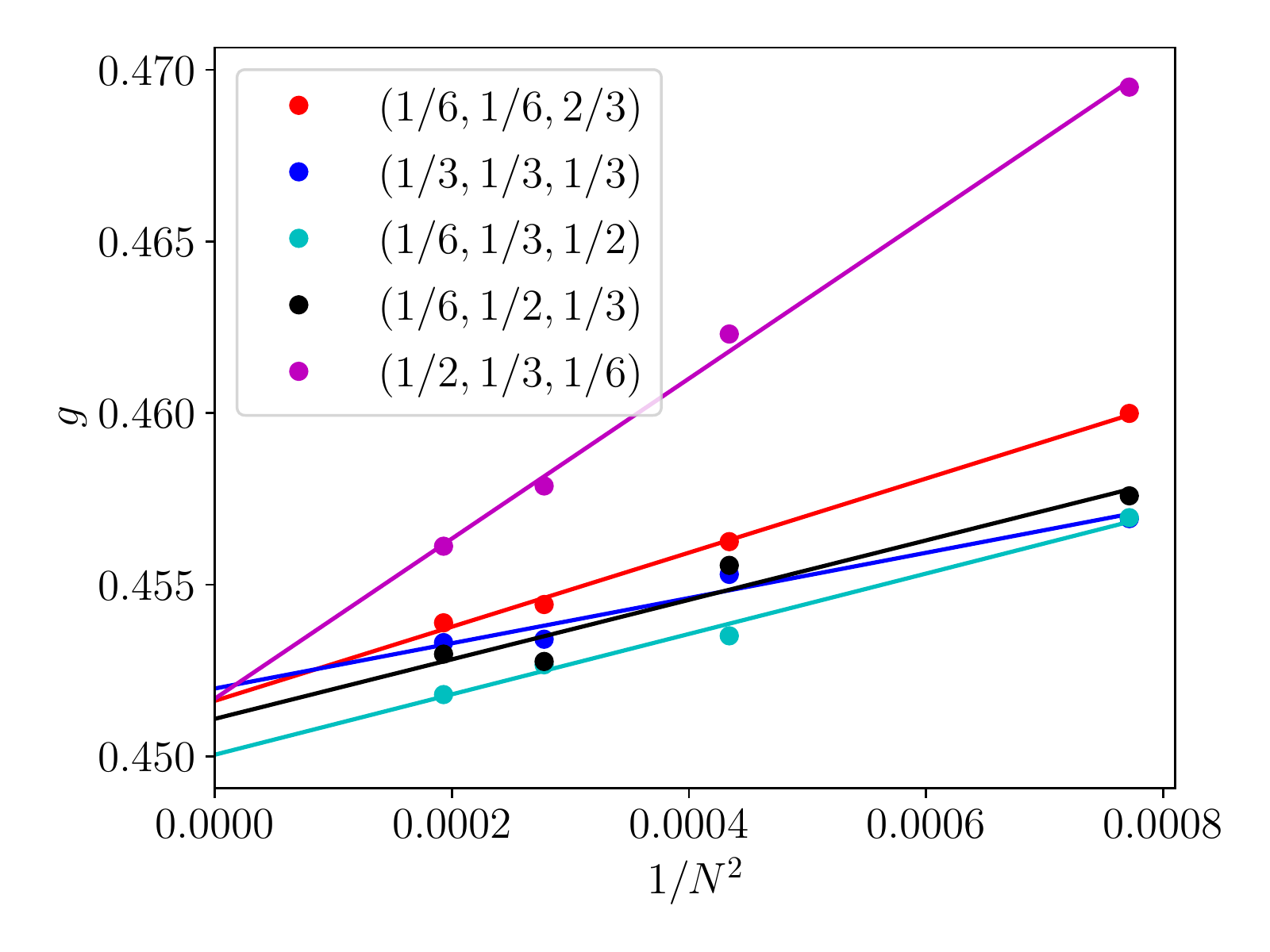}
    \end{minipage}
    \begin{minipage}{0.49\linewidth}
        \subfigure[]{\label{fig:OBFhs}}\\[-1ex]
        \includegraphics[width=\linewidth]{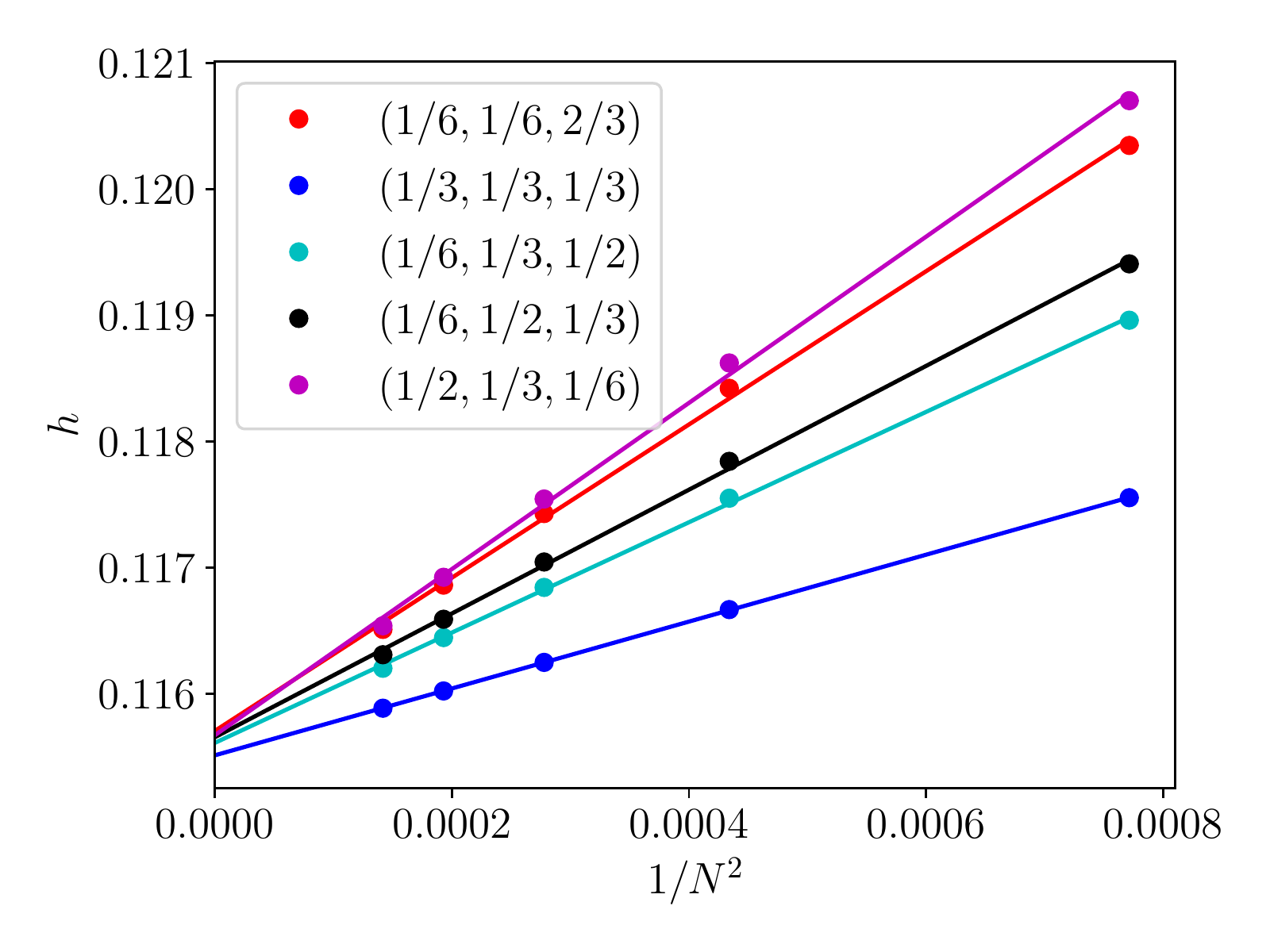}
    \end{minipage}
    \caption{$g(A:B)$ and $h(A:B)$ for the O'Brien-Fendley model~\cite{Obrien_2018} with $\lambda=0.3$ In brackets we show the relative size ratios $(r_A,r_B,r_c)$ of regions $A,B,C$ and we use a sequence of increasing system sizes $36\leq N \leq 84$. We use bond dimensions $18\leq D\leq 26$ and truncated MPS physical dimensions $\tilde{d}_A=\tilde{d}_B=64$ and $\tilde{d}_{C_L}=\tilde{d}_{C_R}=12$. Lines show linear extrapolations of the data points.}
    \label{fig:OBF}
\end{figure}

In Fig.~\ref{fig:OBF}, the intercept of the lines with the vertical axis shows the extrapolation of $g(A:B)$ and $h(A:B)$ to the thermodynamic limit. We see that both $g(A:B)$ and $h(A:B)$ converge to independent values regardless of subregion sizes. The slight differences in the values of extrapolation is caused by linearly fitting the data points with $1/N^2$, whereas the finite-size corrections scale in a more complicated way.

\section{Multipartite entanglement for gapped systems}\label{app:gapped_systems}
In this section, we study $g$ and $h$ for 1D gapped systems in more detail. First and foremost. we assume that the ground state in the thermodynamic limit can be represented as a MPS with finite bond dimension $D$, and that the multipartite entanglement quantities of the ground state can be extracted from the MPS. This assumption is, however, not rigorous proven. Despite tremendous success of infinite MPS algorthms which have been widely used to study general 1D gapped systems, it has only been rigorously proven that local properties can be captured by a MPS with finite bond dimension \cite{Dalzell2019}. Therefore, the argument below should be taken as rigorous only for the gapped systems whose ground state can be exactly represented as a MPS. For a general gapped system, the argument below should be taken as heuristic rather than exact.

\subsection{Fixed-point MPS}
We begin with translation invariant MPS in the thermodynamic limit. Such a state flows toward fixed-point MPS under coarse-graining~\cite{Chen2011, Perez-Garcia2007, Verstraete2005}.  We will show that a fixed-point MPS is a SOTS for any contiguous tripartition. In particular, the MPS is a triangle state if it is injective, which corresponds to no long-range order. Therefore, in the thermodynamic limit the MPS has $h(A:B)=0$, and if the MPS is injective then also $g(A:B)=0$.

We consider a periodic uniform MPS (puMPS) with $N$ sites. Each site has a $d$-dimensional degree of freedom with associated rank-3 tensor $M$ with shape $d\times D \times D$, where $D$ is the bond dimension. We denote by $M_{s_i}$, where $s_i$ indexes the $d$ physical basis states on the $i$-th site, a $D \times D$ matrix. 

In puMPS representation, the many-body ground state may be written in terms of $N$ identical rank-3 tensors $M$:
\begin{align}
    \bigket{\psi(M)} = \sum_{s_1 s_2\cdots s_n}\Tr\big(M_{s_1}M_{s_2}\cdots M_{s_n}\big)\ket{s_1s_2\cdots s_n},
    \label{eq:TI_MPS}
\end{align}
The puMPS representation is invariant under a local similarity transformation $A_i \rightarrow SM_{s_i}S^{-1}$ for all $s_i$ and an invertible $S$. As before, we define a transfer matrix derived from the matrices above
\begin{align}
    T_{\alpha \gamma, \beta \delta} = \sum_s M_{s\alpha \beta} (M_{s\gamma \delta})^*
    \label{eq:T_tensor}
\end{align}
shown graphically in Fig.~\ref{fig:tensor_definitions}. The grouping of the indices indicates that we will treat the four-index tensor formed by the contraction instead as a $D^2 \times D^2$ matrix with legs grouped as  $\alpha\gamma$ and $\beta\delta$. We then denote the product of $n$ adjacent transfer matrices $T^{n}_{\alpha \gamma, \beta \delta}$.

For ground states of gapped spin systems, as $n\rightarrow \infty$, $T^{n}_{\alpha \gamma, \beta \delta}$ approaches a fixed-point transfer matrix $T^{\text{fp}}_{\alpha \gamma, \beta \delta}$. In other words, by coarse-graining more sites, the corresponding transfer matrix converges to a single tensor which represents the renormalization fixed point. Such fixed-point tensors exhibit interesting properties shown in Refs.~\onlinecite{Chen2011} and~\onlinecite{Perez-Garcia2007}. We briefly review those properties and then use them to show that $h(A:B)=0$ for any contiguous tripartition of a MPS in the thermodynamic limit.

Suppose $\ket{\psi}$ is short-range correlated. Correlation functions of observables on two sites separated by $L$ sites must decay to zero as $\exp(-L/\xi)$ where $\xi$ is the correlation length. As observed in Ref.~\onlinecite{Chen2011}, by considering the Jordan normal form of $T_{\alpha \gamma, \beta \delta}$, it can be seen that short-range correlation requires that $T_{\alpha \gamma, \beta \delta}$ must have a non-degenerate largest eigenvalue. By using the similarity transformation, the canonical form introduced in Ref.~\onlinecite{Perez-Garcia2007} can be imposed so that the corresponding right eigenvector of $T_{\alpha \gamma, \beta \delta}$ is $\ket{\Phi_R} = \sum_\alpha \ket{\alpha\alpha}$ and the corresponding left eigenvector is $\ket{\Phi_L} = \sum_\beta \lambda_\beta \ket{\beta\beta}$ where $\sum_\beta \lambda_\beta = 1$. In that case, the fixed-point tensor is given by 
\begin{align}
(T^{\text{fp}})_{\alpha\gamma,\beta\delta} = \ket{\Phi_R}\bra{\Phi_L} = \sum_{\alpha\beta\delta\gamma} \lambda_\beta \delta_{\alpha,\gamma}\delta_{\beta, \delta}
\label{eq:e_infty_zero_g}
\end{align}

In this canonical form, it is easy to read off what the coarse-grained matrices $M^{\text{(fp)}}_{s\alpha\beta}$ defined in Eq.~\eqref{eq:coarse_M} could be by using two physical indices $jL$ and $jR$ instead of just one ($s$):

\begin{align}
M^{\text{(fp)}}_{(jL)(jR)\alpha \beta} &= \delta_{(jL)\alpha} \delta_{(jR)\beta} \sqrt{\lambda_\alpha}
\label{eq:coarse_grained_fp_matrix_SR}
\end{align}

This matrix is shown in graphical tensor notation in Fig.~\ref{fig:fixed_point_M} (a). We may interpret this matrix as follows. The indices $jL$ and $jR$ label basis vectors of the coarse-grained physical Hilbert space $\calH_\text{cg}$ composed of two degrees of freedom $\calH_\text{cg} = \calH_L \otimes \calH_R$ such that $\dim(\calH_\text{cg}) \leq D^2$. Consider a fixed-point MPS composed of matrices $M^{\text{(fp)}}_{(jL)(jR)\alpha \beta}$. By connecting Kronecker deltas of adjacent sites, say sites $k$ and $k^+$, it can be seen that this state is a tensor product of bipartite states shared by $\calH_R$ of site $k$ and $\calH_L$ of site $k^+$ with Schmidt coefficients $\{\sqrt{\lambda_\alpha}\}$. With any contiguous tripartition the fixed-point MPS is clearly a triangle state. Thus $g(A:B)=0$ for short-range correlated MPS in the thermodynamic limit.
\begin{figure}
    \centering
    \includegraphics[width=0.5\linewidth]{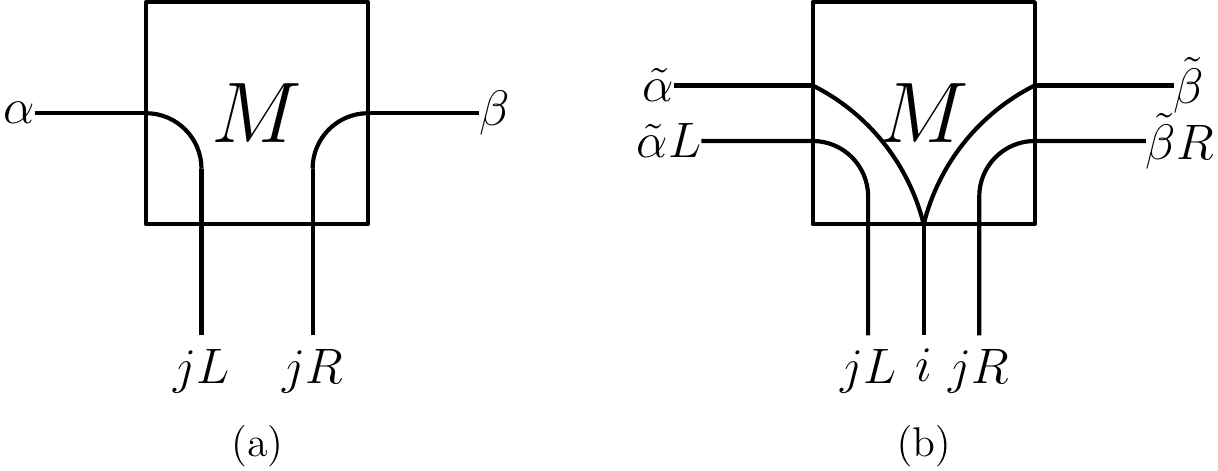}
    \caption{Fixed-point MPS matrices. Solid lines connect indices which are related by a Kronecker delta. (a) $M^{\text{(fp)}}_{(jL)(jR)\alpha \beta}$ for short-range correlated fixed-point states (given in Eq.~\eqref{eq:coarse_grained_fp_matrix_SR}). (b) $M^{\text{(fp)}}_{i(jL)(jR)\tilde{\alpha}(\tilde{\alpha}L)\tilde{\beta}(\tilde{\beta}R)}$ for long-range correlated fixed-point states (given in Eq.~\eqref{eq:coarse_grained_fp_matrix_LR}).}
    \label{fig:fixed_point_M}
\end{figure}

More generally, a long-range correlated state, such as a superposition of macroscopically different ground states, may be described by a transfer matrix with degenerate largest eigenvalue, allowing for some nonzero correlation even at infinite separation. It is shown in Ref.~\onlinecite{Perez-Garcia2007} that by similarity transformations, the $D\times D$ matrices $\{M_{s\alpha \beta}\}$ may be put in block-diagonal form so that each eigenvalue corresponds to an orthogonal subspace of the $D$-dimensional virtual space. Each block may then be put in canonical form as in Eq.~\eqref{eq:e_infty_zero_g}. Suppose the largest eigenvalue is $m$-fold degenerate. The fixed-point transfer matrix then has $m$ blocks satisfying Eq.~\eqref{eq:e_infty_zero_g}. The corresponding generalization of Eq.~\eqref{eq:coarse_grained_fp_matrix_SR} can be seen most easily by introducing a physical index $i$ which indexes the orthogonal subspace and using a pair of labels instead of just one for each virtual index $\alpha \rightarrow (\tilde{\alpha}, \tilde{\alpha}L)$ and $\beta \rightarrow (\tilde{\beta}, \tilde{\beta}R)$:

\begin{align}
M^{\text{(fp)}}_{i(jL)(jR)\tilde{\alpha}(\tilde{\alpha}L)\tilde{\beta}(\tilde{\beta}R)} &= \delta_{i\tilde{\alpha}} \delta_{\tilde{\alpha}\tilde{\beta}} \delta_{(jL)(\tilde{\alpha}L)} \delta_{(jR) (\tilde{\beta} R)} \sqrt{\lambda_{\tilde{\alpha} (\tilde{\alpha}L)}}
\label{eq:coarse_grained_fp_matrix_LR}
\end{align}

For a fixed value of $i=\alpha=\beta$, then, the matrix reduces to the form of Eq.~\eqref{eq:coarse_grained_fp_matrix_SR} and the interpretation as a product of bipartite states holds. We may then interpret the overall matrix in Eq.~\eqref{eq:coarse_grained_fp_matrix_LR} to mean that the many-body state is sum of $m$ states, each described by Eq.~\eqref{eq:coarse_grained_fp_matrix_SR}. A many-body state defined by such matrices is therefore a SOPS. With any contiguous tripartition the fixed-point state is thus a SOTS and $h(A:B)=0$.

In conclusion, we have shown that a translation invariant MPS flows to a fixed-point MPS which is a SOPS in the thermodynamic limit. By Lemma~\ref{lemma:SOP_CG}, such a state has $h(A:B)=0$ for all contiguous tripartitions. If the MPS is short-range correlated, then the fixed-point MPS is further simplified to a triangle state, which has $g(A:B)=h(A:B)=0$ for all contiguous tripartitions.

\subsection{Translation invariant MPS at finite sizes}
One can go further to bound $g(A:B)$ and $h(A:B)$ where the sizes of $A$ and $B$ are taken to be finite. The aim of this section is to show that they are exponentially close (in terms of the lengths of $A$ and $B$) to the fixed-point values. This essentially follows from the continuity of $E_P(A:B)$, $S_R(A:B)$ and $I(A:B)$ with respect to the density matrices. Let $\rho_{AB}$ and $\sigma_{AB}$ be two density matrices on the Hilbert space $\calH_A\otimes \calH_B$, where the dimensions of $\calH_{A/B}$ is $d_{A/B}$. Let $\Delta(\rho_{AB}, \sigma_{AB}) = (1/2)|\rho_{AB}-\sigma_{AB}|$ denote the trace distance between $\rho_{AB}$ and $\sigma_{AB}$ where $|A| = \Tr\sqrt{A^\dagger A}$. Then if $\Delta(\rho_{AB}, \sigma_{AB})  \leq \epsilon$ and for $\epsilon$ sufficiently small we have the three theorems below. The continuity of $g(A:B)=2E_P(A:B)-I(A:B)$ and $h(A:B)=S_R(A:B)-I(A:B)$ is then implied by the theorems. 
\begin{theorem}[Continuity of $E_P$ \cite{Terhal2002}]\label{thm:continuity_ep}
 $|E_P(\rho_{AB})-E_P(\sigma_{AB})|\leq 40 \sqrt{\epsilon} \log d - 4 \sqrt{\epsilon} \log (4 \sqrt{\epsilon}) $, where $d=d_A d_B.$
\end{theorem}
\begin{theorem}[Continuity of $S_R$ \cite{Akers2020}]\label{thm:continuity_sr}
 $|S_R(\rho_{AB})-S_R(\sigma_{AB})|\leq 4 \sqrt{2\epsilon} \log(\min\{d_A.d_B\})-2\sqrt{2\epsilon}\log \epsilon $.
\end{theorem}
\begin{theorem}[Continuity of Mutual Information \cite{Terhal2002}]\label{thm:continuity_I} $|I(\rho_{AB})-I(\sigma_{AB})|\leq 3 \epsilon \log d  - 3 \epsilon \log\epsilon$, where $d=d_A d_B$.
\end{theorem}

To be more concrete, we consider a puMPS of $N$ sites and bond dimension $D$ and take each of the three regions $A, B, C$ to be of size $N/3$. We use the coarse-graining of the MPS to make a puMPS on three sites (Eq.~\eqref{eq:truncated_MPS}), where each site represents the coarse-grained Hilbert space of $A,B$ and $C$. Note that here  no truncation on the physical Hilbert space is used, so the physical dimension of each site is $d_A = d_B = d_C = D^2$ and the three tensors in Eq.~\eqref{eq:truncated_MPS} are the same. The puMPS on three sites are related to the original state by local isometries, so the coarse-graining itself does not change any entanglement properties, including $E_P(A:B), S_R(A:B)$ and $I(A:B)$. We now show that the coarse-grained MPS is exponentially close in $N$ to the fixed-point MPS on three sites, so by continuity of $g(A:B)$  and $h(A:B)$, they are also exponentially close to the their fixed-point values ($g^{\text{fp}}(A:B)=0$ for injective MPS and $h^{\text{fp}}(A:B)=0$ regardless of injectivity).

First, we assume the MPS is injective and derive a bound on $g(A:B)$. The transfer matrix has a unique eigenvalue $1$. Denote the second-largest eigenvalue as $\lambda_2<1$. The correlation length is then $\xi = -1/\log \lambda_2$. The injectivity of the MPS is then equivalent to finite correlation length. The transfer matrix has an eigenvalue decomposition
\begin{equation}
T_{\alpha\gamma, \beta\delta} =T^{\text{fp}}+ \lambda_2 r_{\alpha\gamma} l_{\beta\delta} + \cdots
\end{equation}
where $r$ and $l$ are the right/left eigenvectors of the eigenvalue $\lambda_2$ and $\cdots$ denotes contributions of smaller eigenvalues. Taking large powers of $T_{\alpha\gamma, \beta\delta}$, the $\cdots$ term vanishes faster than the second term, so we will drop the dots. Then we have
\begin{equation}
T^{N/3}_{\alpha\gamma, \beta\delta} =T^{\text{fp}}+ e^{-N/(3\xi)} r_{\alpha\gamma} l_{\beta\delta}.
\end{equation} 
We will measure the difference in terms of the norm $||A_{abc...}|| = \sqrt{A_{abc...}A^{*}_{abc...}}$, where repeated indices are summed. Then
\begin{equation}
\label{eq:TNexpansion}
||T^{N/3}_{\alpha\gamma, \beta\delta}- T^{\text{fp}}|| = e^{-N/(3\xi)} \sqrt{\Tr(r^{\dagger} r) \Tr(l^{\dagger} l)},
\end{equation}
which decays exponentially with system size $N$. Denote the coarse-grained tensor on $A,B,C$ as $M$, then
\begin{align}
\label{eq:Mexpansion}
    T^{N/3}_{\alpha \gamma, \beta \delta} = \sum_s M_{s\alpha \beta} (M_{s\gamma \delta})^*.
\end{align}
Recall that $M$ can be obtained by an eigenvalue decomposition Eq.~\eqref{eq:TMn_MPS} and \eqref{eq:coarse_grain_eigv}. One can use the Rayleigh-Schrodinger perturbation theory to derive the difference between $M$ and $M^{\text{fp}}$. Notice that the differences in the eigenvalues and eigenvectors are of order $e^{-N/(3\xi)}$, and the combination Eq.~\eqref{eq:coarse_grain_eigv} at most change on the order of $e^{-N/(6\xi)}$ because of the square root in the eigenvalues.
%The difference can be bounded using Eqs.~\eqref{eq:TNexpansion}, and the final expression is
%\begin{equation}
%||M - M^{\text{fp}}|| \leq D \sqrt{Be^{-N/(3\xi)}\left(1+\frac{B e^{-N/(3\xi)}}{\Delta^2}\right)},
%\end{equation}
%where $\Delta$ is the smallest gap of the Schmidt spectrum $\{\sqrt{\lambda_\alpha}\}$ (excluding degeneracy). At large sizes, the second term in the bracket is negligible, then
Then at large sizes
\begin{equation}
||M - M^{\text{fp}}|| \leq O(1) \cdot D e^{-N/(6\xi)},
\end{equation}
Let $\ket{\psi_N}$ be a puMPS with tensor $M$ on 3 sites, and $\ket{\psi^{\text{fp}}}$ be a puMPS with tensor $M^{\text{fp}}$ on 3 sites, then
\begin{equation}
\Delta(\ket{\psi_N}\bra{\psi_N}, \ket{\psi^{\text{fp}}}\bra{\psi^{\text{fp}}}) \leq  O(1) \cdot D e^{-N/(6\xi)}.
\end{equation}
Finally, let $\rho_{AB}=\Tr_C |\psi_N\rangle\langle \psi_N|$ and $\sigma_{AB}=\Tr_C|\psi^{\text{fp}}\rangle\langle \psi^{\text{fp}}|$. Since the trace distance is monotonic under tracing out a subsystem, 
\begin{equation}
\Delta(\rho_{AB},\sigma_{AB}) \leq O(1)\cdot D e^{-N/(6\xi)}. 
\end{equation}
Finally we can derive a bound on $g(A:B)$ and $h(A:B)$,
\begin{equation}
g(A:B)\leq 2|E_P(\rho_{AB})-E_P(\sigma_{AB})|+|I(\rho_{AB})-I(\sigma_{AB})|
\end{equation}
and 
\begin{equation}
h(A:B)\leq |S_R(\rho_{AB})-S_R(\sigma_{AB})|+|I(\rho_{AB})-I(\sigma_{AB})|.
\end{equation}
Upon using the continuity theorems \ref{thm:continuity_ep}, \ref{thm:continuity_sr}, and \ref{thm:continuity_I}, we see that both $g(A:B)$ and $h(A:B)$ are upper bounded by an exponentially decaying quantity.

For a general MPS we can again decompose it into a sum of superselection sectors which are locally orthogonal. The coarse-graining transformation acts on each of the superselection sectors separately. At finite sizes the state is coarse-grained into a SOTS with an expoentially small correction and therefore $h(A:B)$ is upper bounded by an exponentially small quantity.
\subsection{MPS without translation invariance}
The requirement of translation invariance above is not essential. As noted in Refs. \onlinecite{Chen2011, Perez-Garcia2007}, the coarse graining can be done in a similar way for MPS without translation invariance. Here we briefly review how this is done. Denote the tensor on site $k$ as $M^{(k)}$ and the corresponding transfer matrix as $T^{(k)}$. It has been shown in Ref. \onlinecite{Perez-Garcia2007} that an injective MPS can be put into the central canonical form, where the (unique) dominant eigenvalue of $T^{(k)}$ is $1$ and the corresponding left/right eigenvectors are $\Lambda^{(k)}_{\alpha\gamma} = \sqrt{\lambda^{(k)}_{\alpha}}\delta_{\alpha\gamma}$ and $\Lambda^{(k+1)}_{\beta\delta} = \sqrt{\lambda^{(k+1)}_{\beta}}\delta_{\beta\delta}$, where $\sum_{\alpha}\lambda^{(k)}_{\alpha}=1$. Note that the left dominant eigenvector of $T^{(k+1)}$ is the same as the right dominant eigenvector of $T^{(k)}$. The transfer matrices have the eigenvalue decomposition
\begin{equation}
T^{(k)}_{\alpha\gamma, \beta\delta} = \Lambda^{(k+1)}_{\beta\delta} \Lambda^{(k)}_{\alpha\gamma} + \lambda^{(k)}_2 r^{(k)}_{\beta\delta} l^{(k)}_{\alpha\gamma} + \cdots.
\end{equation}
We assume that the transfer matrix has a finite gap, $\lambda^{(2)}_k<1-\epsilon', \forall k$ for some $\epsilon'>0$. This is equivalent to exponentially decaying correlation functions typical in gapped systems. The coarse-graining amounts to multiplying the transfer matrices in an interval together. If the interval is long enough, then the only remaining part is the multiplication of the first term in the expansion. This gives the coarse-grained tensor on the left of Fig.~\ref{fig:fixed_point_M}. The resulting state is then a triangle state. Similarly, in the case of a generic MPS one can decompose it into a sum of injective MPS and then apply the coarse-graining to give a SOTS.
\end{document}